\documentclass[12pt, draftclsnofoot, onecolumn]{IEEEtran}
\usepackage{amsmath, lipsum}
\usepackage{listings}
\usepackage{indentfirst}
\usepackage{graphicx}
\usepackage{float}
\usepackage{extarrows}
\usepackage{amssymb}
\usepackage{amsthm}
\usepackage{cuted}
\usepackage{epstopdf}
\usepackage{caption}
\usepackage{subcaption}
\usepackage{cite,color}
\usepackage{todonotes}
\usepackage{tikz,pgfplots}
\usepackage{xcolor}
\usetikzlibrary{automata,positioning}
\usetikzlibrary{arrows,shapes,chains}
\usepgflibrary{patterns}
\usepackage{algorithm}
\usepackage{algorithmic}
\usepackage{tikz,pgfplots}
\usetikzlibrary{arrows, intersections}
\pgfplotsset{plot coordinates/math parser=false}

\newtheorem{lemma}{Lemma}

\newtheorem{theorem}{Theorem}
\newtheorem{definition}{Definition}

\newtheorem{proposition}{Proposition}
\newtheorem{remark}{Remark}
\newtheorem{assumption}{Assumption}
\usepackage{enumitem}
\setlist[itemize]{leftmargin=*}

\begin{document}
	
	\title{Real-time Sampling and Estimation on Random Access Channels: Age of Information and Beyond}
	\date{}
	\author{Xingran Chen,\IEEEmembership{}
		Xinyu Liao, \IEEEmembership{} 
		Shirin Saeedi Bidokhti \IEEEmembership{}
		
		\IEEEcompsocitemizethanks {\IEEEcompsocthanksitem Xingran Chen and Shirin Saeedi Bidokhti are with the Department of Electrical and System Engineering,  University of Pennsylvania, PA, 19104.\quad 
			E-mail: \{xingranc, saeedi\}@seas.upenn.edu.
			\IEEEcompsocthanksitem Xinyu Liao is with the Graduate Group of Applied Mathematics and Computational Science, University of Pennsylvania, PA, 19104.\quad 
			E-mail: xinyul@sas.upenn.edu.
		}
		\thanks{This work was supported by Grants NSF 1910594 and NSF 1850356.}
		\thanks{Parts of the work have been accepted by 2021 IEEE INFOCOM.}
	}
	\maketitle
	\begin{abstract}
		Next generation  multiple access channels require to provision for  unprecedented massive user access  in a plethora of  applications in cyber-physical systems. This work proposes decentralized policies for the real-time monitoring and estimation of  autoregressive processes over random access channels. Two classes of policies are investigated: (i) oblivious schemes in which  sampling and transmission policies are independent of the  processes that are monitored, and (ii) non-oblivious schemes in which transmitters causally observe their corresponding processes for decision making. In the class of oblivious policies, we show that minimizing  the expected time-average estimation error is equivalent to minimizing the expected age of information. Consequently, we prove lower and upper bounds on the minimum achievable estimation error in this class. Next, we consider non-oblivious policies and design a threshold policy, called error-based thinning, in which each transmitter node becomes active  if its instantaneous error has crossed a fixed threshold (which we optimize). Active nodes then transmit stochastically following a slotted ALOHA policy. A closed-form, approximately optimal, solution is found for the threshold as well as the resulting estimation error.  It is shown that non-oblivious policies offer a multiplicative gain close to $3$ compared to oblivious policies. Moreover, it is shown that oblivious policies that use  age of information for decision making improve the state-of-the-art at least by the multiplicative factor $2$. The performance of all discussed policies is compared using simulations. Numerical comparison shows that the performance of the proposed decentralized policy is very close to that of centralized greedy scheduling.  Finally, we extend our framework to unreliable random access channels. Simulations show that the multiplicative gain offered by non-oblivious policies (compared to oblivious policies) is independent of the channel erasure probability.
	\end{abstract}
	
	\section{Introduction}\label{sec:intro}
	\subsection{Motivation}
	The Internet of Things (IoT) paradigm is changing our
	conception of communications.
	In the past decades, research has focused on various  technologies to improve connectivity, rate, reliability, and/or latency  \cite{SAA2015, LYQZDZ2020, PPLLY2019, VBPSSSRM2019, AASEAA2019,  LCNZYCFYGW2021,PLJX2020, XZWC2019,JC2016,TWASAVWH2013, DCRHGHY2016,LHEUMJ2019,TPGLWSMB2021, JDZYXX2009, SRFAKSSQZA2020,LLBSYWDHKT2021, KWZDDKCSGKK2021, SRFAKSSQZA2017,JZTLRPLXSYJGNCB2019, JLLJCH2008,CXJZWNXSRPLTLJW2019,A.Kosta-2019,8807257,I.Kadota-2019-1,A.M.Bedewy-2019,I.Kadota-2018,I.Kadota-2019}.
	But are such metrics representative of optimal multiple access system designs for future IoT and cyber-physical systems (CPS) applications or do we need to go beyond these metrics? 
	Expanding on this question, it is important to note that in traditional designs, it is often assumed that information bits (or packets) are processed and stored at the sources, waiting to be reliably transmitted and replicated at the receiver node(s) with high rate and low latency. 
	However, in many IoT applications, these assumptions are no longer realistic. Oftentimes,
	information is to be collected and communicated real-time. In such settings, rate, reliability and latency may not directly be relevant.

	In this paper, we consider the  problem of {real-time sampling and estimation} over a random access channel with $M$ Markov (autoregressive) physical processes (sources).  These processes are to be observed, sampled, and communicated wirelessly  with a fusion center for \emph{timely estimation}. Considering applications in IoT and CPS, it is not realistic to assume
	a central scheduler that monitors all the sensors for decision making; we therefore seek to design  near optimal decentralized sampling and communication strategies.

	Towards understanding the problem of real-time sampling and estimation in the above setting, recent works \cite{RTEM2019, SacBacinogluUysal-BiyikogluDurisi2018, cxritw2019, cxr2019, YSun-2017} have  proposed centralized and decentralized multiple access schemes to minimize the metric of Age of Information (AoI) \cite{S.Kaul2011-1, S.Kaul2011-2}.
	Nonetheless, it has remained open how such designs can inform near-optimal designs when the desired metric of performance is the real-time estimation error (rather than AoI as a proxy). 
	The present work establishes this bridge and goes beyond AoI minimization.

	\subsection{Related Work}
	Below, we discuss three major facets of the problem.
	
	{\bf{Sampling}}: 
	Remote estimation of physical processes requires efficient sampling and communication strategies that minimize not only the estimation error cost but also the sampling and transmission costs. With this viewpoint, prior works have studied optimal sampling strategies and their structural properties for various point-to-point scenarios. \cite{ImerBasar10} designs optimal sampling strategies with limited measurements. \cite{RabiMoustakidesBaras12} studies the problem for continuous sources.  \cite{LipsaMartins11} proves the joint optimality of symmetric thresholding policies and Kalman-like estimators for  autoregressive Markov processes. \cite{AMSH2017} formulates  a two-player team problem and designs efficient iterative algorithms. Systems with energy harvesting sensors are considered in \cite{Nayyar-Basar-Teneketzis-Veeravalli12}. Noisy channels and packet drop channels are considered in \cite{ChakravortyMahajan2020, GaoAkyolBasar2018}.  The above-mentioned works have all considered single-user channels and the developed methodologies do not generalize to random access networks with multiple sensors.
	
	{\bf{Reliable v.s. Timely  communication}}: 
	In estimation and control applications, timeliness of communication is key and that is why traditional rate-distortion frameworks and channel coding paradigms that propose asymptotic block coding solutions are not applicable. More importantly, it is oftentimes observed that as the  rate and/or reliability of compression/communication schemes improve, their timeliness decrease. This aspect of sampling and remote estimation is  barely  studied in the estimation literature. One of the few existing works in this direction is \cite{KH2019} which proposes and optimizes a hybrid automatic repeat request (HARQ)-based remote estimation protocol and  improves the performance of the remote estimation systems compared to  conventional non-HARQ policies.  Recently, tradeoffs between reliability/rate and the timeliness of communication have been looked at in the context of age of information (AoI) -- a metric of timeliness defined in \cite{S.Kaul2011-1}. In channels with queue constraints, \cite{RTEM2019} establishes a tradeoff between AoI and rate. \cite{SacBacinogluUysal-BiyikogluDurisi2018} finds the optimal blocklength of channel coding for minimizing AoI.  
	\cite{cxritw2019} provides a  centralized scheduling framework to attain  tradeoffs between rate and AoI in broadcast channels. \cite{cxr2019} proposes decentralized transmission strategies for random access channels that benefit from the availability of fresh packets and improve both communication rate and AoI.  It is known that AoI is closely related to the expected estimation error of schemes that are oblivious to the processes they monitor \cite{YSun-2017}. Non-oblivious sampling schemes are, however, signal-dependent and known to outperform oblivious schemes. In \cite{YSun-2017}, threshold policies are shown to be optimal for point-to-point channels with a random delay and closed form solutions are found for the optimal threshold value. It is further shown that the oblivious policies can be far from optimal. We build on our prior work in \cite{cxr2019} that concerned AoI minimization and propose decentralized threshold policies for minimizing estimation error in random access channels with many users.
	
	{\bf{Distributed decision making}}: 
	In random access networks, a large number of sensors communicate with a single fusion center over a wireless channel. To avoid collision, most works in this direction have considered centralized oblivious policies  that do not observe the process realizations for decision making (see, e.g., \cite{Meier,Oshman94,Logothetis99,Gupta06,Vitus10,Zhao14,STJSS2015} and the references therein).  
	In   IoT applications, however, it is not realistic to assume a central scheduler that monitors all the sensors for decision making. We seek decentralized solutions in which each sensor decides when to sample and transmit  information based only on its local observations.
	In decentralized setups (and in the context of control, rather than estimation) \cite{KG2015}, \cite{KG2015-2} consider
	wireless control architectures with multiple control loops over a random access channel and optimize the access rate of the sensors who randomly communicate. Policies that adapt to the state of the systems are proposed in \cite{GT2012}.  The work \cite{XZMMVWCUM2021} (which was carried out concurrently and independently) designs decentralized  policies for the remote estimation of i.i.d processes over a collision channel. Decision making in both \cite{GT2012} and \cite{XZMMVWCUM2021} is thresholding and based on the realization of the process (or a function of that). But since neither of the two works exploits channel collision feedback, adaptations of them (or other policies that impose fixed transmission probabilities  on the channel) are far from optimal in our setup.

	\subsection{Contributions}
	We study sampling and remote estimation of $M$ independent random walk processes over a \emph{wireless collision channel}. As opposed to all prior works, we seek decentralized solutions in which  decision at each node is based solely on its local observations and channel collision feedback. Our goal is to minimize the estimation error, and specifically a  normalized metric that we call  the normalized average estimation errors (NAEE). This metric looks at the expected time-average estimation error, normalized by the number of source nodes $M$. We are interested in the asymptotic regime where $M\to\infty$.
	
	Two general classes of policies are considered, namely oblivious policies and non-oblivious policies. In the former class, decision making is independent of the processes that are monitored and we prove that minimizing the expected time-average estimation error, in the class of oblivious policies, is equivalent to minimizing the age of information. This leads to lower and upper bounds on the minimum  achievable estimation error in this class along with efficient oblivious policies that are age-based. In particular, the NAEE under age-based policies  is lower bounded by $.88\sigma^2$ and upper bounded by $\frac{e}{2}\sigma^2$.
	
	We next ask if non-oblivious policies can provide a significant gain by observing the processes as they progress. Since all source nodes are provided with the channel collision feedback, they can compute their age-function and reproduce their respective estimated processes (at the destination) in each time slot. Furthermore, using the collision feedback, the nodes can implicitly coordinate for communication. We  define the notion of  {\it error process} at each node which is a function of the sample values and age. We then propose  a threshold policy, called error-based thinning, in which source nodes become active only when their corresponding error process is beyond a given threshold. Once a node becomes active, it transmits stochastically following a slotted ALOHA policy.

	To find an optimal threshold and find a closed-form solution for the resulting NAEE, we first provide a closed-form expression for the  NAEE that is a function of the peak age, the transmission delay,  a term which we call the silence delay, as well as the process realization. We approximately find the NAEE under an optimal threshold policy by considering the underlying autoregressive Markov process as a discretized Wiener process.  An optimal threshold is then shown to be approximately $\sigma\sqrt{eM}$ and the resulting NAEE to be $\frac{e}{6}\sigma^2$. The approximation error increases linearly as a function of the variance of the innovation process and decreases as $M$ gets large. 
	
	Next, we extend our framework to unreliable random access channels. The closed-form of NAEE under oblivious and non-oblivious policies are provided. The multiplicative gain is the same as that in reliable random access channels, which equals to $3$, and is independent of the channel erasure probability $\epsilon$. Under non-oblivious policies, an optimal threshold is approximately $\sigma\sqrt{eM/(1-\epsilon)}$ and the corresponding NAEE is $\frac{e}{6(1-\epsilon)}\sigma^2$.
	
	Simulation results show that the proposed decentralized  threshold policy outperforms  oblivious  policies. Moreover, oblivious policies are shown to outperform all state-of-the-art policies (both oblivious and non-oblivious) that impose a fixed rate (without using the collision feedback). Finally, it is numerically shown that the performance of the optimal threshold policy is  very close to that of centralized greedy policies that schedule transmissions according to the instantaneous error reduction or age reduction.

	The paper is organized as follows. In Section \ref{sec:systemModel}, we introduce the system model. Oblivious policies are studied in Section~\ref{sec: Signal-independent policies} and non-oblivious policies are discussed in  Section~\ref{sec: signal-dependent}. The framework is extended to unreliable random access channels in Section~\ref{sec: generalization}. Simulation results are presented for various policies  in Section~\ref{sec: numerical results} and our assumptions and derivations are verified numerically. Finally, we conclude in Section~\ref{sec: Conclusion and Future Work}.

	\subsection{Notation}
	We use the notations $\mathbb{E}[\cdot]$ and $\Pr(\cdot)$ for expectation and probability, respectively. Scalars are denoted by  lower case letters, e.g. $s$, and random variables are denoted by capital letters, e.g. $S$. The notation $A\sim B$ implies that $A$ has the same distribution as $B$ and $\mathcal{N}(0, \sigma^2)$ stands for  the Gaussian distribution with mean $0$ and variance $\sigma^2$. The notations $O(\cdot)$ and $o(\cdot)$ represent the Big O and little o notations according to Bachmann-Landau notation, respectively.

	\section{System Model}
	\label{sec:systemModel}
	Consider a system with $M$ statistically identical sensors and a fusion center. We often refer to the sensor nodes as nodes or transmitters and the fusion center as the receiver/destination. For analysis tractability, we focus on symmetric cases where all nodes have similar configuration. Let time be slotted. Each node $i,\ i=1,2,\cdots, M$, observes a process $\{X_i(k)\}_{k\geq0}$ which is a random walk process as follows 
	\begin{align}\label{eq: Wiener process}
		X_i(k+1)=X_i(k)+W_i(k)
	\end{align}
	where $W_i(k)\sim\mathcal{N}(0,\sigma^2)$.  The processes $\{X_i(k)\}_{k=0}^{\infty}$ are assumed to be mutually independent across $i$ and for simplicity we let $X_i(0)=0$.
	
	At the beginning of each time slot, the nodes have the capability to sample the underlying process and decide whether or not to communicate the sample with the receiver. The communication medium is modeled  by a collision channel: If two or more nodes transmit  in the same time slot, then the packets interfere with each other (collide) and do not get delivered at the receiver. We use the binary variable $d_i(k)$  to indicate whether a packet is transmitted from node $i$ and delivered at the receiver in time slot $k$. Specifically, $d_i(k)=1$ if node $i$ delivers a packet successfully; $d_i(k)=0$ otherwise.
	We assume a delay of one time unit in delivery for packets. At the end of time slot $k$, all transmitters are informed (through a low-rate feedback link) whether or not collision occurred, which is indicated by an indicator $c(k)$. If collisions happen in time slot $k$, then $c(k)=1$; if a packet is delivered successfully at the receiver or no packet is transmitted, then $c(k)=0$.  Note that if $c(k)=1$, then $d_i(k)=0$; if $c(k)=0$, then $d_i(k)=0$ when node $i$ does not transmit, and $d_i(k)=1$ when node $i$ transmits a packet. 
	
	We assume that the buffer size of every transmitter is one packet and that new packets replace older undelivered  packets at the transmitter. This assumption relies on the fact that the underlying processes that are monitored are Markovian.

	The receiver estimates the process in every time slot based on the  collection of the received samples. Denote by $\hat{X}_i(k)$ the estimate of $X_i(k)$ in time slot~$k$.
	We define the following normalized average sum of estimation errors (NAEE) as our performance metric:
	\begin{equation}\label{eq: E-MMSE}
		\begin{aligned}
		L^{\pi}(M)=\lim_{K\rightarrow\infty}\mathbb{E}[L_K^{\pi}],\quad L_K^{\pi}(M)=\frac{1}{M^2}\sum_{i=1}^{M}\frac{1}{K}\sum_{k=1}^{K}\big(X_i(k)-\hat{X}_i(k)\big)^2
		\end{aligned}
	\end{equation}
	where $M$ is the number of sources, $\pi\in\Pi$ refers to the sampling and transmission policy in place, and $\Pi$ is the set of all decentralized sampling and transmission policies.  Note that the metric \eqref{eq: E-MMSE} is normalized by $M$. This allows us to study the asymptotic performance in the regime of large $M$.
	The minimum attainable NAEE is then denoted by $L(M)$:
	\begin{align}\label{eq:Lgeneral}
		L(M)=\min_{\pi\in\Pi}L^{\pi}(M).
	\end{align}
	Our objective is to design \emph{decentralized} sampling and transmission mechanisms  to attain $L(M)$.  
	
	Consider the $i^{th}$ node. Let $\{k_\ell^{(i)}\}_{\ell\geq 0}$ be the sequence of time slots at the end of which packets are received at the destination from node $i$. In any time slot $k$, $k_{\ell-1}^{(i)}<k\leq k_\ell^{(i)}$, the latest sample from node $i$ is received at the end of $k_{\ell-1}^{(i)}$ and since collisions may happen, then it is time stamped at the beginning of time $k'$ with $k'\leq k_{\ell-1}^{(i)}$.  So the age of information (AoI) \cite{cxr2019} with respect to node $i$, denoted by $h_i(k)$, is 
	\begin{align}\label{eq: h_i}
		h_i(k)=k-k_{\ell-1}^{(i)}.
	\end{align}
	Without loss of generality, assume $k_0^{(i)}=0$.
	At the beginning of time slot $k$, the receiver knows the information of all packets delivered before time $k$, i.e., $\{X_i(k_t^{(i)})\}_{t=0}^{\ell-1}$
	and reconstructs $\hat{X}_i(k)$ by the minimum mean square error (MMSE) estimator:
	\begin{align*}
		\hat{X}_i(k)=&\mathbb{E}\left[X_i(k)|\big\{X_i\big(k_t^{(i)}\big)\big\}_{t=0}^{\ell-1}\right].
	\end{align*}
	For the class of policies that we consider in this paper (oblivious policies and symmetric thresholding policies), the MMSE estimator reduces to a Kalman-like  estimator:
	\begin{equation}\label{eq: MMSE}
		\begin{aligned}
			\hat{X}_i(k)=&\mathbb{E}[X_i(k)|X_i(k_{\ell-1}^{(i)})]=X_i(k_{\ell-1}^{(i)}).
		\end{aligned}
	\end{equation}

	One of the major challenges in this problem arises from the decentralized nature of decision making. A decentralized policy is one in which the action of each node is only a function of its own local observations and actions. In this setup,  the action of node $i$ at time $k$ depends on the history of feedback and actions as well as casual observations of the process $\{{X}_i(j)\}_{j\leq k}$.
	
	We also consider a simpler class of policies $\Pi^\prime$, called \emph{oblivious} policies, in which the action of each node depends only on the history of feedback and actions at that node. In particular, oblivious policies do not take into account the realization (value) of the samples, but only the time they were sampled, transmitted, and received (if successfully received). 
	We denote the minimum attainable NAEE in the class of oblivious policies by 
	\begin{align}\label{eq:Lgeneral1}
		\bar{L}(M)=\min_{\pi\in\Pi^\prime}L^{\pi}(M).
	\end{align}
	We argue in section~\ref{sec: Signal-independent policies} that this simplification equivalently transforms the estimation problem  into the problem of timely communication of packets for age minimization. By additionally exploiting the  value of the samples, in Section \ref{sec: signal-dependent}, we design and analyze decentralized mechanisms that outperform oblivious schemes in minimizing the expected average estimation error.

\section{Oblivious Policies and Age of Information}\label{sec: Signal-independent policies}

Oblivious policies are independent of the processes they observe and they are therefore less costly to implement. Moreover, they can still benefit from the channel collision feedback to (i) quantify how stale the information at the receiver has become (in order to decide when to sample and communicate) and (ii) adapt to the channel state (for communication purposes). In this section, we show that minimizing  NAEE in the class of oblivious policies is equivalent to minimizing the normalized average sum of AoI (NAAoI) as we have previously defined in \cite{cxr2019}.

First, we establish the following relationship between  the expected estimation error and the expected age.
\begin{lemma}\label{lem: Gaussian distribution}
	In oblivious policies, the expected  estimation error associated with process $i$ has the following relationship with the expected age function:
	\begin{align}\label{eq: Gaussian distribution}
		\mathbb{E}[\big(X_i(k)-\hat{X}_i(k)\big)^2]=\mathbb{E}[h_i(k)]\sigma^2.
	\end{align}
\end{lemma}
\begin{proof}
	At the beginning of time slot $k$, the estimation error is
\begin{align*}
X_i(k)-\hat{X}_i(k)=X_i(k)-X_i(k_{\ell-1}^{(i)})=\sum_{l=1}^{k-k_{\ell-1}^{(i)}}W_i\big(l+k_{\ell-1}^{(i)}\big).
	\end{align*}
	By the stationarity of $\{W_i(k)\}_{k=1}^{\infty}$ and using \eqref{eq: h_i}, we conclude
	\begin{align*}
		X_i(k)-\hat{X}_i(k)\sim\sum_{l=1}^{h_i(k)}W_i(l).
	\end{align*}
	Now note that $h_i(k)$ is independent of $\{W_i(k)\}_{k=1}^{\infty}$ under oblivious policies. Therefore, using Wald's equality, we find
	\begin{align*}
		&\mathbb{E}[X_i(k)-\hat{X}_i(k)]=0\\
		&\mathbb{E}[\big(X_i(k)-\hat{X}_i(k)\big)^2]=\mathbb{E}[h_i(k)]\sigma^2.
	\end{align*}
\end{proof}
\begin{remark}
	Lemma \ref{lem: Gaussian distribution} does not hold for  non-oblivious policies. As a matter of fact, finding $\mathbb{E}[\big(X_i(k)-\hat{X}_i(k)\big)^2]$ in closed-form is non-trivial and its numerical computation can be intractable when $M$ is large. The reason is that even though the estimation error is the sum of  $h_i(k)$ Gaussian noise variables,  once we condition on $h_i(k)$, their distributions change because $h_i(k)$ can be dependent on the process that is being monitored.
\end{remark} 
Lemma~\ref{lem: Gaussian distribution} is reminiscent of \cite[Lemma~4]{TZOYS2019}.
Using Lemma~\ref{lem: Gaussian distribution}, the metric NAEE in \eqref{eq: E-MMSE} can be re-written as follows:
\begin{align}\label{eq: obliviousL}
	L^\pi(M)=\lim_{K\rightarrow\infty}\sigma^2 J^\pi(M)
\end{align}
where
\begin{align}\label{eq: obliviousJ}
	J^\pi(M)=\frac{1}{M^2}\sum_{i=1}^{M}\frac{1}{K}\sum_{k=1}^{K}\mathbb{E}[h^\pi_i(k)].
\end{align}
Note that $J^\pi(M)$ is only a function of the age function $h^\pi_i(k)$. The metric in \eqref{eq: obliviousJ} is the NAAoI defined in \cite{cxr2019} and, therefore,
the decentralized threshold policies of \cite{cxr2019} apply directly. Note that the generation rate of packets for every sensor is $\theta\in(0,1]$ in \cite{cxr2019}, while $\theta$ should be set to $1$ in the model defined in Section~\ref{sec:systemModel}. This is because  we assume that  sensor $i$ can observe the process $\{X_i(k)\}_k$ for every $k$.
In particular,   \cite[Algorithm 2]{cxr2019}  outlines a  stationary age-based thinning (SAT) policy  in which a source transmits only when the corresponding AoI is larger than a pre-determined threshold. Using this algorithm, it was shown that the following age performance can be achieved in the limit of large $M$:
\begin{align} 
	&\lim_{M\to\infty}J^{\text{SAT}}(M)=\frac{e}{2}\label{eq: upperboundJ}\\
	&\lim_{M\to\infty}L^{\text{SAT}}(M)=\frac{e}{2}\sigma^2\label{eq: upperboundL}.
\end{align}
Results from \cite[Proposition 1]{cxr2019} also lead to the following lower bound on  NAAoI $J^\pi(M)$ for any decentralized policy $\pi$:
\begin{align}
	\label{eq: lowerboundJ}
	\lim_{M\to\infty}J^\pi(M)\geq .88.
\end{align}
Using \eqref{eq: upperboundL} and \eqref{eq: lowerboundJ}, we arrive at the following proposition.
\begin{proposition}\label{pro: oblivious}
	The minimum attainable NAEE in the class of oblivious policies is characterized by the following bounds
	\begin{align}\label{eq: SAT}
		.88\sigma^2\leq \lim_{M\rightarrow\infty}\bar{L}(M)\leq\frac{e}{2}\sigma^2. 
	\end{align}
\end{proposition}

\subsection{Comparison with Oblivious Centralized Policies}
In this section, we compare the SAT policy in \cite[Algorithm 2]{cxr2019} with an oblivious centralized policy -- the Max-Weight (MW) policy \cite{cxritw2019, cxr2019, I.Kadota-2018, I.Kadota-2019-1, IKEM2019}. 
Denote $\underline{T}^{(i)}(k)=\{k_j^{(i)}\}_{j=0}^{\ell}$ with $k_\ell^{(i)}\leq k$. We devise the MW policy using techniques from Lyapunov Optimization.  Define the Lyapunov function
\begin{align}\label{eq: Lyapunov}
L(k)=\frac{1}{M}\sum_{i=1}^{M}\big(X_i(k)-\hat{X}_i(k)\big)^2
\end{align}
and the one-slot Lyapunov Drift
\begin{align}\label{eq: L Drift}
LD(k)=\mathbb{E}[L(k+1)-L(k)|\underline{T}^{(i)}(k)].
\end{align}
We devise the MW policy such that it minimizes the one-slot Lyapunov Drift. 
\begin{definition}\label{def: MW policy}
	At the beginning of each slot $k$, the MW policy chooses the action $i^*$ such that
	\begin{align}
		h_{i^*}(k)=\max_i h_i(k).
	\end{align}
\end{definition}
Note that this  policy is exactly the MW policy derived in \cite{IKEM2019} for age minimization. From Lemma~2 in \cite[Section~III]{I.Kadota-2018}, the policy defined in Definition~\ref{def: MW policy} is optimal. 
\begin{proposition}\label{thm: MW policy}
	The MW policy in Definition~\ref{def: MW policy} minimizes the one-slot Lyapunov Drift in each slot, and
	\begin{align}\label{eq: MW error}
		\lim_{M\rightarrow\infty}L^{MW}(M)=\frac{\sigma^2}{2}.
	\end{align}
\end{proposition}
\begin{proof}
	The proof of Proposition~\ref{thm: MW policy} is given in Appendix~\ref{App: Proof MW}.
\end{proof}
Comparing \eqref{eq: upperboundL} with \eqref{eq: MW error}, we have  $$\lim_{M\rightarrow\infty}\frac{L^{SAT}(M)}{L^{MW}(M)}=e.$$
The NAEE of the decentralized SAT policy is $e$ times that of the optimal centralized policy in the limit of large $M$. The conclusion coincides with one's intuition: the throughput of the decentralized SAT policy in \cite{cxr2019} is $e^{-1}$, while the throughput of the centralized MW policy is $1$, which implies the amount of delivered fresh packets in the centralized MW policy is $e$ times that of the decentralized SAT policy.
We illustrate their performances through simulations in Section~\ref{sec: numerical results}.

	\section{Non-oblivious Policies}\label{sec: signal-dependent}
	We now consider a more general class of policies in which the nodes can observe their corresponding Markov processes for decision making. In other words, we seek to  benefit from not only the AoI, but also the process realization (in a casual manner). Clearly, if all nodes try to transmit their samples at every time slot, no packet will go through due to collisions. The nodes, therefore, need to transmit packets with a lower rate. This means that they have to decide, in a decentralized manner, when to transmit. Motivated by the optimality of threshold policies in various point-to-point setups  \cite{ImerBasar10, RabiMoustakidesBaras12, Nayyar-Basar-Teneketzis-Veeravalli12, YSun-2017}, as well as their applications in age minimization over many-to-one random access channels  \cite{cxr2019}, we propose threshold policies for decision making.

	\subsection{Error-based Thinning}\label{sec: error-Based Thinning}

	Define the  {\it error process} $\psi_i(k)$ at node $i$ as follows: 
	\begin{align}\label{eq: markov process}
		\psi_i(k) = |X_i(k)-\hat{X}_i(k)|.
	\end{align}
	Since the transmitters have access to collision feedback, they can calculate $\hat{X}_i(k)$, and hence $\psi_i(k)$, in each time slot and use this information for decision making.
	One way to understand $\psi_i(k)$ is as follows. At time $k$, if the sample of node $i$ is successfully delivered, the estimation error will reduce by $\psi_i(k)$. So $\psi_i(k)$ quantifies the amount of instantaneous estimation error reduction upon successful delivery from transmitter $i$. With this viewpoint, we devise a threshold policy in which transmitters prioritize packets that have large $\psi_i(k)$.
	In particular, we design a fixed threshold $\beta$ to distinguish and prioritize nodes that offer a high instantaneous gain.
	
	The action of each node is thus as follows: node $i$ becomes ``active" if the error process $\psi_i(k)$ has crossed a pre-determined threshold $\beta.$
	Once a transmitter is active, it remains active until a packet is successfully delivered from that node.   Active nodes transmit stochastically following  Rivest’s stabilized slotted ALOHA protocol \cite[Chapter~4.2.3]{textbookslottedaloha}. Denote the number of active nodes and an estimate of the number of active nodes in time slot $k$ as $N(k)$, $\hat{N}(k)$, respectively. In particular, each active node transmits its sample with probability $p_b(k)$ which is calculated adaptively as follows based on an estimate of the number of active nodes\footnote{Since the sensors have unit buffer sizes, the
		number of  ``backlogged" nodes $N(k)$ in Rivest’s algorithm is at most $M$. One notes that  this has been incorporated in \eqref{eq: slotted aloha}. }:
	\begin{equation}\label{eq: slotted aloha}
		\begin{aligned}
			&p_b(k)=\min(1,\frac{1}{\hat{N}(k)})\\
			&\hat{N}(k)=\left\{\begin{aligned}
				&\hat{N}(k-1)+\hat{\lambda}(k)+(e-2)^{-1}\quad\text{if}\,\, c(k-1)=1\\
				&\hat{\lambda}(k)+\left( \hat{N}(k-1)-1\right)^+ 
				\quad\text{if}\,\, c(k-1)=0.
			\end{aligned}\right.
		\end{aligned}
	\end{equation}
	Here, $\hat{\lambda}(k)$ is an estimate of $\lambda(k)$, and $\lambda(k)$ is the sum arrival rate in time slot $k$. It is well-known that the maximum sum throughput of the slotted ALOHA is $e^{-1}$  \cite[Chapter~4.2.3]{textbookslottedaloha} and the regime of interest is $\lambda(k)<e^{-1}$ when $k$ is sufficient large. In our setup, $\lambda(k)$ corresponds to the {\it expected} number of nodes that become {\it active} in time slot $k$ (see Definition~\ref{def: active nodes} ahead). We refer to $\lambda(k) $  as the activation rate or the effective arrival rate in time slot $k$.
	
	So far, we have outlined a threshold policy in which a node decides to become active  if its local error process is larger than a pre-determined threshold value $\beta$. We call this procedure  {\it Error-based Thinning} (EbT). The main underlying challenge is, however, in the design of the {\it optimal} $\beta$. In the rest of this section, we will find an approximately optimal choice for $\beta$  and analyze the corresponding NAEE. We start by some preliminaries.
	\subsection{Preliminaries}

	Consider node $i$ and an inter-delivery interval $(k_{\ell-1}^{(i)}, k_\ell^{(i)}]$ (see Figure~\ref{fig:JUL}). The inter-delivery  time $I_\ell^{(i)}$ is  given by $I_\ell^{(i)}=k_\ell^{(i)}-k_{\ell-1}^{(i)}$. For any time slot $k$, $k_{\ell-1}^{(i)}<k\leq k_{\ell}^{(i)}$, we can write the error process $\psi(k)$ as follows:
	\begin{align}\label{eq: difference XhatX}
		\psi_i(k)=|X_i(k)-\hat{X}_i(k)|=\Big|\sum_{j=k_{\ell-1}^{(i)}}^{k-1}W_i(j)\Big|.
	\end{align}
	Note that from 	\eqref{eq: h_i}, $h_i(k)=k-k_{\ell-1}^{(i)}$, the term on the right hand side of \eqref{eq: difference XhatX} is the sum of $h_i(k)$ independent Gaussian noise variables. Indeed, \eqref{eq: difference XhatX} demonstrates that $\psi_i(k)$ contains both the information of sample values as well as the age with respect to source $i$. 
	
	We next define  ``active" nodes as follows.
	\begin{definition}[Active Nodes]\label{def: active nodes}
		If there exists a time slot $k_0\in(k_{\ell-1}^{(i)}, k_\ell^{(i)}]$ such that (i) $\psi_i(j)<\beta$ for all $k_{\ell-1}^{(i)}<j<k_0$ and (ii) $\psi_i(k_0)\geq\beta$, then we say that node $i$ is {\it active} in the entire interval $[k_0,k_\ell^{(i)}]$. 
	\end{definition}
	\begin{definition}[Silence Delay and Transmission Delay]\label{def: silience delay}
		Let $k_0$ be as defined in Definition \ref{def: active nodes}. We define  $J_\ell^{(i)}=k_0-k_{\ell-1}^{(i)}$ as the silence delay, and  $U_\ell^{(i)}=k_\ell^{(i)}-k_0+1$  as the transmission delay (see Figure~\ref{fig:JUL}).
	\end{definition}
	\noindent  An active source becomes inactive immediately after a successful delivery.  By the above two definitions, the inter-delivery time $I_\ell^{(i)}$ consists of two components -- the silence delay $J_\ell^{(i)}$ and the transmission delay $U_\ell^{(i)}$:
	\begin{align}\label{eq: IlLlUl}
		I_\ell^{(i)}=J_\ell^{(i)}-1+U_\ell^{(i)}.
	\end{align}
	In this equation, $J_\ell^{(i)}$ is the first time slot after $k_{\ell-1}^{(i)}$ at which $\psi_i(k)\geq\beta$ (as defined in Definition \ref{def: silience delay}).
	So $J_\ell^{(i)}-1$ represents the number of time slots in which node $i$ is  not active, and $U_\ell^{(i)}$ represents the number of time slots in which node $i$ is in active state. Recall that  active nodes transmit with probability $p_b(k)$. So $U_\ell^{(i)}$ may be larger than $1$ either because the node is active and it does not transmit or because the node transmits and experiences collision.  
	By the stationarity of the transmission scheme, the processes $\{I_\ell^{(i)}\}_{i,\ell}$, $\{J_\ell^{(i)}\}_{i,\ell}$, and $\{U_\ell^{(i)}\}_{i,\ell}$ are statistically identical across $i$ and $\ell$. 
	We define $I_\beta$, $J_\beta$, and $U_\beta$ to have the same distributions as $\{I_\ell^{(i)}\}_{i,\ell}$, $\{J_\ell^{(i)}\}_{i,\ell}$, and $\{U_\ell^{(i)}\}_{i,\ell}$, respectively. 
	
	\begin{figure}[t!]
		\centering
		\includegraphics[width=2.4in]{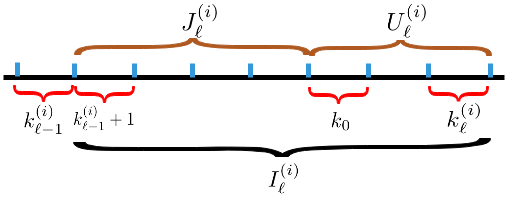}
		\caption{an example of $J_\ell^{(i)}$, $U_\ell^{(i)}$, and $I_\ell^{(i)}$.  Packets are generated at the beginning of every time slot, so $J_\ell^{(i)}$ arrivals/generations means $J_\ell^{(i)}-1$ time slots.}
		\label{fig:JUL}
	\end{figure}

	Let $\{W_j\}_j$ be an i.i.d sequence with the same distribution as $\{W_j(k)\}_j$.
	Define $$S_n=\sum_{j=1}^{n}W_j.$$
	Using the definition of $h_i(k)$ in \eqref{eq: h_i}, and by the stationarity of $\{W_j\}_j$, we conclude that
	\begin{align}\label{eq: S-E}
		\psi_i(k)\sim |S_{h_i(k)}|.
	\end{align}
	Recall that $J_\beta$ has the same distribution as  $J_\ell^{(i)}$. Then,  $J_\beta$  is the smallest time index at which $|S_n|\geq\beta$ in an inter-delivery interval.
	$J_\beta$ is a stopping time for $S_n$. From \cite[Chapter~7.5.1, Lemma~7.5.1]{Gallagerbook}, it follows that $J_\beta$  has finite moments of all orders.  Moreover, using \cite[Chapter~7.5.2]{Gallagerbook}, we have
	\begin{align}\label{eq: EJ}
		\mathbb{E}[S_{J_\beta}^2]=\sigma^2\mathbb{E}[J_\beta].
	\end{align}

	Finding an optimal $\beta$ is non-trivial because $\beta$ impacts both $J_\beta$ and $U_\beta$. 
	In the remainder of this subsection, we establish some useful expressions for the expectations of $I_\beta$ and $U_\beta$ in an optimal design.
	
	Let $a(k)$ denote the number of {\it newly} active nodes at time $k$, i.e., the number of nodes that become active from inactive states. We have $\mathbb{E}[a(k)]=\lambda(k)$, where
	$\lambda(k)$ is the expected sum arrival rate in time slot $k$ (imposed by our sampling and transmission policy). Now recall that in a traditional slotted Aloha-based random access channel,  the maximum  sum throughput  is asymptotically $e^{-1}$. This is true also for the case with buffer size $1$ where only the latest packets are stored,  as discussed in \cite[Appendix~E]{cxr2019}) and which applies to our setting here. Define $c(M)$ as the sum rate/throughput when the system contains $M$ sources.
	\begin{definition}\label{def: stabilized}
		The random access system is stabilized\footnote{Here, contrary to traditional slotted ALOHA schemes, the term ``stabilized''  does not refer to  ``stability of queues" in our problem setup. However, the term ``stabilized'' implies that the system is stationary, when the sum arrival rate is  less than $e^{-1}$.} if $\lambda_m = \limsup_{k\to\infty}\lambda(k)<e^{-1}$.
	\end{definition}
	\noindent We provide our analysis under the following two assumptions:
	\begin{assumption}\label{assumption: independent}
		Under an optimal $\beta$, when $M$ is sufficiently large, $\{a(k)\}_{k=1}^{\infty}$ are approximately independent. 
	\end{assumption}
	\begin{assumption}\label{assumption: fully channel capacity}
		Under an optimal $\beta$, when $M$ is sufficiently large, the random access system is stabilized, and $\lambda_m\approx e^{-1}$, $c(M)\approx e^{-1}$. 
	\end{assumption}
	
	Assumptions~\ref{assumption: independent},~\ref{assumption: fully channel capacity} are given for analysis tractability, and we will verify them  for our proposed $\beta$ later.  In the rest of the paper, let $M$ be sufficiently large. We seek to find an optimal $\beta$ under assumptions~\ref{assumption: independent},~\ref{assumption: fully channel capacity}.
	To transmit as many fresh samples as possible, $\beta$ is designed such that $\lambda(k)$ is as large as possible. Thus, we focus on the regime where $\lambda(k)$ is close to $e^{-1}$ when $k$ is large, and from Assumption~\ref{assumption: fully channel capacity}, $\lambda_m\approx e^{-1}$. For tractability in analysis,  we let the estimate $\hat{\lambda}(k) = e^{-1}$ for all $k$.
	Specifically, we replace $\hat{\lambda}(k)$ with $e^{-1}$ in \eqref{eq: slotted aloha}.

	Note that the system is stationary, so $U_\ell^{(i)}$ (or $U_\beta$) is a random variable and measurable. Recall from \cite[Chapter~7.5.1, Lemma~7.5.1]{Gallagerbook}, $J_\beta$ has finite moments of all orders. Therefore, $I_\beta$ is measurable. Now, we first show that the (strong) law of large numbers holds for $\{I_\ell^{(i)}\}_\ell$.
	We remark that while  $\{I_\ell^{(i)}\}_\ell$ is not independent, it is weakly correlated across $\ell$,  as we prove in Appendix~\ref{App: proof of I LLN}. We can thus conclude that the strong law of large numbers holds for $\{I_\ell^{(i)}\}_\ell$, see also \cite{RLyons}.

	Recall that $N(k)$ is the number of active nodes  at the {\it beginning} of time slot $k$. The fraction of active nodes at the {\it beginning} of time slot $k$ is hence  $N(k)/M$. 
	\begin{definition}
		Define $\alpha_\beta(k)$ as the expected fraction of active nodes:
		\begin{align}\label{eq: alpha def}
			\alpha_\beta(k)=\frac{\mathbb{E}[N(k)]}{M}.
		\end{align}
	\end{definition}
	\noindent If $\beta=0$, then all nodes are active and $\alpha_0(k)=1$; if $\beta=+\infty$, then all nodes are inactive and $\alpha_{+\infty}(k)=0$. 
	In the limit of $k\rightarrow\infty$, we denote the expected fraction of  active nodes by~$\alpha_\beta$:
	\begin{align}\label{eq: limit of alpha(k)}
		\alpha_\beta=\lim_{k\rightarrow\infty}\frac{\mathbb{E}[N(k)]}{M}=\lim_{k\to\infty} \mathbb{E}\left[\frac{1}{M}\sum_{i=1}^M \mathbf{1}(\text{node $i$ is active at time $k$})\right]
	\end{align}
	The limit in \eqref{eq: limit of alpha(k)} exists because the transmission  policy is stationary and  hence the sequence in the expectation above is stationary in the steady state. Continuing from \eqref{eq: limit of alpha(k)}, we  have
	\begin{equation}\label{eq:getridofi}
	\begin{aligned}
			\alpha_\beta&=\lim_{K\to\infty} \mathbb{E}\left[\frac{1}{MK}\sum_{k=1}^K\sum_{i=1}^M \mathbf{1}(\text{node $i$ is active at time $k$})\right]\\
			&\stackrel{(a)}{=}\mathbb{E}\left[\lim_{K\to\infty} \frac{1}{MK}\sum_{k=1}^K\sum_{i=1}^M \mathbf{1}(\text{node $i$ is active at time $k$})\right].
		\end{aligned}
	\end{equation}
	where step $(a)$ holds by the dominated convergence theorem because the sequence in the expectation \eqref{eq:getridofi} is a fraction and bounded by~$1$.
	Utilizing the symmetry and stationarity with respect to various nodes (the system), we prove the following lemma in Appendix~\ref{App: proof of alpha}, signifying that
	$\alpha_\beta$ represents the fraction of time that each node is active in the limit of $K\to\infty$, hence represents the probability of each node being active when the system is steady.
	
	\begin{lemma}\label{lem: alpha}
		When the system is stabilized, $\alpha_\beta$ exists,  and
		$\alpha_\beta=\frac{\mathbb{E}[U_\beta]}{\mathbb{E}[I_\beta]}$.
	\end{lemma}

	Since $\alpha_\beta$ exists, then, when $k\to\infty$, the expected number of nodes that become active in every time slot is $(1-\alpha_\beta)M\alpha_\beta$, and
	\begin{align}\label{eq: alpha_solution}
		(1-\alpha_\beta)M\alpha_\beta=\lim_{k\to\infty}\lambda(k)=\limsup_{k\to\infty}\lambda(k)= \lambda_m.
	\end{align}	
	From Assumption~\ref{assumption: fully channel capacity}, $\lambda_m\approx e^{-1}<1$. Using \eqref{eq: alpha_solution}, one sees that $M\alpha_\beta$ is an infinitesimal of higher order than $M$. Now using Lemma~\ref{lem: alpha}, we can show that $\mathbb{E}[U_\beta]$ is an infinitesimal of higher order than $M$, as discussed in the following lemma.
	\begin{lemma}\label{lem EU}
		When the system is stabilized, 
		\begin{align}
			&\mathbb{E}[I_\beta]= \frac{M}{c(M)}\label{eq: I_beta}\\
			&\mathbb{E}[U_\beta]= \frac{M}{c(M)}\alpha_\beta = o(M) \label{eq: U}
		\end{align}
		where $\alpha_\beta$ is the expected fraction of active nodes in the steady state as defined in~\eqref{eq: limit of alpha(k)}. 
	\end{lemma}
	\begin{remark}
		Lemma~\ref{lem EU} coincides with one's intuition. Recall that the throughput of the channel is $c(M)$, so the throughput for each node is $\frac{c(M)}{M}$ (due to the symmetry). From the perspective of expectation, every successful delivery takes $\frac{M}{c(M)}$ time slots, i.e., $\mathbb{E}[I_\beta]= \frac{M}{c(M)}$. In addition, note that the expected number of active node is $M\alpha_{\beta}$, so the throughput of every active  node is $\frac{c(M)}{M\alpha_{\beta}}$. Again, from the perspective of expectation, every successful delivery from active nodes takes $\frac{M}{c(M)}\alpha_{\beta}$ time slots, i.e., $\mathbb{E}[U_\beta]= \frac{M}{c(M)}\alpha_{\beta}$.
	\end{remark}
	\begin{proof}
		The proof of Lemma~\ref{lem EU} is given in Appendix~\ref{App: proof of EU EI}. 
	\end{proof}

	\subsection{The closed form of NAEE}
	We next  derive a closed form expression for the attained NAEE, $L^{EbT}(M)$.
	Using \eqref{eq: S-E}, we re-write \eqref{eq: E-MMSE} as follows.
	\begin{equation}\label{eq: E-MMSE-1}
		L^{EbT}(M)=\lim_{K\rightarrow\infty}\mathbb{E}[\frac{1}{M^2K}\sum_{i=1}^{M}\sum_{k=1}^{K}S_{h_i(k)}^2].
	\end{equation}
	Define $\Delta^{(i)}_\ell$ as the sum of $S_{h_i(k)}^2$ in the interval $k\in(k_{\ell-1}^{(i)}, k_\ell^{(i)}]$:
	\begin{align}\label{eq: def delta}
		\Delta^{(i)}_\ell=\sum_{k=k_{\ell-1}^{(i)}+1}^{k_{\ell}^{(i)}}S_{h_i(k)}^2.
	\end{align}
	Since $h_i(k)$ has the same distribution in the interval $[k_{\ell-1}^{(i)}+1, k_{\ell}^{(i)}]$ over $i$ and $\ell$, then $\Delta_\ell^{(i)}$ has the same distribution over $i$ and $\ell$. We define $\Delta_\beta$ to have the same distribution as $\Delta_\ell^{(i)}$. The next lemma shows that the expected time average in \eqref{eq: E-MMSE-1} takes a closed form expression in terms of $\mathbb{E}[\Delta_\beta]$ and $\mathbb{E}[I_\beta]$. 
	
	\begin{lemma}\label{lem: I delta SSN}
		The proposed EbT policy attains the following NAEE:
		\begin{align}\label{eq: LLN L}
			L^{EbT}(M)=\frac{1}{M}\frac{\mathbb{E}[\Delta_\beta]}{\mathbb{E}[I_\beta]}.	
		\end{align}
	\end{lemma}
	
	\begin{proof}
		The proof of Lemma~\ref{lem: I delta SSN} is given in Appendix~\ref{App: LLN for sequences}. 
	\end{proof}

	Similar to \eqref{eq: def delta}, $\Delta_\beta$ can be expressed as
	\begin{align}\label{eq: delta_beta}
		\Delta_\beta = \sum_{j=1}^{I_\beta} S_j^2.
	\end{align}
	From \eqref{eq: delta_beta}, the NAEE in \eqref{eq: LLN L} can now be re-written as follows
\begin{equation}\label{eq: L-delta1}
\begin{aligned}
L^{EbT}(M)=\frac{1}{M}\frac{\mathbb{E}\Big[\sum_{j=1}^{I_\beta}S_j^2\Big]}{\mathbb{E}[I_\beta]}=\frac{1}{M}\frac{\mathbb{E}\Big[\sum_{j=1}^{J_\beta+U_\beta-1}S_j^2\Big]}{\mathbb{E}[I_\beta]}\triangleq L_1^{EbT}(M)+L_2^{EbT}(M)
\end{aligned}
\end{equation}
where
\begin{align}
L_1^{EbT}&(M)=\frac{1}{M}\frac{\mathbb{E}\Big[\sum_{j=1}^{J_\beta}S_j^2\Big]}{\mathbb{E}[I_\beta]}\label{eq: Lpi221}\\
L_2^{EbT}&(M)=\frac{1}{M}\frac{\mathbb{E}\Big[\sum_{j=J_\beta+1}^{J_\beta+U_\beta-1}S_j^2\Big]}{\mathbb{E}[I_\beta]}=\frac{1}{M}\frac{2\mathbb{E}[J_\beta](\mathbb{E}[U_\beta]-1)+\mathbb{E}[U_\beta^2]-\mathbb{E}[U_\beta]}{2\mathbb{E}[I_\beta]}\sigma^2\label{eq: Lpi222}.
\end{align}
The equality in \eqref{eq: Lpi222} is proved in Appendix~\ref{App: proof of L2}. Note that $L^{EbT}$ is a function of the peak age $I_\beta$, the silience delay $J_\beta$,  the transmission delay $U_\beta$, and the process realization through~$W_j$.

	\subsection{Optimizing $\beta$ Approximately}\label{sec: approximation of beta}

	Finally, we find approximate closed form expressions for $L_1^{EbT}(M)$ and $L_2^{EbT}(M)$. Let $M$ be sufficiently large. Using \eqref{eq: U} along with the  the fact that  $\mathbb{E}[J_\beta]\leq\mathbb{E}[I_\beta]$, one can simplify \eqref{eq: Lpi222}  in the limit of large $M$:
	\begin{align}\label{eq: Lpi223}
		L_2^{EbT}(M)=&\frac{1}{M}\cdot\frac{\mathbb{E}[U_\beta^2]}{2\mathbb{E}[I_\beta]}\sigma^2.
	\end{align}

	The following lemma comes in handy in our approximations.
	\begin{lemma}\label{lem: brown motion}
		Consider a Brownian motion $B_t$. Define $J=\inf\{t\geq0, |B_t|\geq a\}$. The following holds:
		\begin{itemize}
			\item [] (1) \cite[Chapter~7, Theorem~7.5.5, Theorem~7.5.9]{Durrett} $\mathbb{E}[J] = a^2$ and $\mathbb{E}[J^2]=\frac{5a^4}{3}$;
			\item [] (2) $\mathbb{E}[\int_{0}^{J}B_t^2dt]=\frac{1}{10}\mathbb{E}[J^2] = \frac{1}{6}a^4.$
		\end{itemize}
	\end{lemma}
	\begin{proof}
		The proof of Lemma~\ref{lem: brown motion} is given in Appendix~\ref{App: Proof of brown motion}.
	\end{proof}
	
	For any $j$, $\frac{S_j}{\sigma}$ is  Gaussian with mean zero and variance $j$. We propose to use $B_j$ as an approximation of $\frac{S_j}{\sigma}$. Letting $a=\beta/\sigma$ in Lemma~\ref{lem: brown motion}, we obtain
	\begin{align}
		&\mathbb{E}\big[J_\beta]\approx\frac{\beta^2}{\sigma^2},\,\, \mathbb{E}[J_\beta^2]\approx \frac{5\beta^4}{3\sigma^4}\label{eq: J}\\
		&\mathbb{E}\big[\sum_{j=1}^{J_\beta}S_j^2\big]\approx\frac{\beta^4}{6\sigma^2}\approx\frac{1}{10}\mathbb{E}[J_\beta^2]\label{eq: sumJ}.
	\end{align}
	The approximation error analysis is provided in Section~\ref{sec: Approximating Error Analysis}.
	
	Substituting \eqref{eq: sumJ} into \eqref{eq: L-delta1}, we  find the following approximation for $L^{EbT}$:
	\begin{align}\label{eq: L1 2}
		\hat{L}^{EbT}(M)= \frac{\frac{1}{5}\mathbb{E}[J_\beta^2] + \mathbb{E}[U_\beta^2]}{2M\mathbb{E}[I_\beta]}\sigma^2.
	\end{align}

	\begin{theorem}\label{thm: optimalbeta}
Let $M$ be sufficiently large. The optimal $\beta^*$ is approximately given by
\begin{align*}
\beta^*\approx\hat{\beta}= \sigma\sqrt{eM},
\end{align*}
		and
		\begin{align}\label{eq: estimate L}
			\hat{L}^{EbT} = \frac{e}{6}\sigma^2.
		\end{align}
	\end{theorem}
	\begin{proof}
		The detailed proof of Theorem~\ref{thm: optimalbeta} is given in Appendix~\ref{App: Proof of optimal beta}. Here, we only provide a roadmap of the proof. 
		{\bf (i)} After simplifying $\hat{L}^{EbT}(M)$ in \eqref{eq: L1 2} by using \eqref{eq: IlLlUl}, \eqref{eq: I_beta}, \eqref{eq: U}, \eqref{eq: J}, we find
		$\hat{L}^{EbT}(M)\approx \frac{\frac{M^2}{c(M)^2} - \frac{2\beta^4}{3\sigma^4}+Var(U_\beta)}{2\frac{M^2}{c(M)}}\sigma^2.$
		{\bf (ii)} We  show that the term $\frac{Var(U_{\beta^*})}{2\frac{M^2}{c(M)}}$ is negligible. {\bf (iii)} We derive $\beta^*\approx\hat{\beta}=\sigma\sqrt{\frac{M}{c(M)}}$ as an (approximate) minimizer of NAEE. This leads to $\hat{L}^{EbT} = \frac{1}{6c(M)}\sigma^2$. 
		
	\end{proof}
	
	Finally, Assumptions~\ref{assumption: independent},~\ref{assumption: fully channel capacity}  are verified (approximately) for $\beta^*$  when $M$ is sufficiently large in Appendix~\ref{App: assumptions verify}.

	It is interesting to compare the performance of the proposed EbT policy with the  oblivious decentralized and centralized policies of Section \ref{sec: Signal-independent policies}. 
	From \eqref{eq: obliviousL}, \eqref{eq: obliviousJ}, and \eqref{eq: upperboundJ},
	$$
	\lim_{M\to\infty}L^{SAT}(M)=\frac{e}{2}\sigma^2.
	$$
	using \eqref{eq: upperboundL} and \eqref{eq: estimate L}, we obtain
	\begin{align}\label{eq: penalty3}
		\lim_{M\to\infty}\frac{L^{SAT}(M)}{\hat{L}^{EbT}(M)}\approx3.
	\end{align}
	The NAEE of oblivious SAT policy is around three times that of the EbT policy.
	From \eqref{eq: SAT}, the NAEE of the oblivious MW policy of Section \ref{sec: Signal-independent policies} is asymptotically $\frac{\sigma^2}{2}$ and comparing with $\frac{e}{6}\sigma^2=0.455\sigma^2$ one concludes that the NAEE of the EbT policy is close to that of the oblivious MW policy. We remark that since $\hat{L}^{EbT}(M)$ is an estimate of $L^{EbT}(M)$, these comparisons are not exact. We will also compare the numerical performance of Algorithm \ref{alg: Limit Threshold-slotted ALOHA} with oblivious policies as well as other state-of-the-art algorithms in Section \ref{sec: numerical results}. 	
	Algorithm~\ref{alg: Limit Threshold-slotted ALOHA} below summarizes the proposed decentralized error-based transmission policy.
	\begin{algorithm}
		\caption{Error-based Thinning (EbT)}\label{alg: Limit Threshold-slotted ALOHA}
		\begin{algorithmic}
			\STATE {Set the time horizon $K$.}
			\STATE {Set initial points:  $h_i(0)=1$, $X_i(0)=\hat{X}_i(0)=0$ for $i=1,2,\cdots,M$; $c(0)=0$; $d_i(k)=0$; $p_b(0)=1$; $n(0)=0$; $k=1$. }

			\STATE {Set ${\tt \beta}^*=\sigma\sqrt{eM}$.} 
			
			\REPEAT
			
			\STATE {\bf Step 1:} For each node $i$, observe the collision feedback $c(k-1)$ and $d_i(k-1)$ at the end of time slot $k-1$, and update $k_\ell^{(i)}$'s and $\hat{X}_i(k)$, respectively.

			\STATE {\bf Step 2:} For each node $i$, observe $X_i(k)$ \big(which evolves according to \eqref{eq: Wiener process}\big) and compute $\psi_i(k)$ by \eqref{eq: markov process}.
			
			\STATE {\bf Step 3:}  If $\psi_i(k)<{\tt \beta}^*$, then node $i$ does not transmit packets; otherwise it transmits a packet with probability $p_b(k)$.
			
			\STATE {\bf Step 4:} Calculate $p_b(k)$ by \eqref{eq: slotted aloha} in which $\lambda(k)=e^{-1}$.	
			
			\UNTIL{$k=K$}
			
			\STATE Calculate $$L_K^{EbT}=\frac{1}{M^2}\sum_{i=1}^{M}\frac{1}{K}\sum_{k=0}^{K}\psi_i^2(k).$$
		\end{algorithmic}
	\end{algorithm}

	\subsection{Approximation Error Analysis}\label{sec: Approximating Error Analysis}
	Note that approximations are used in \eqref{eq: J} and \eqref{eq: sumJ}, now we analyze the approximation error in terms of $\sigma^2$.  The approximation error of $L^{EbT}$ consists of (i) the approximation error in \eqref{eq: J} and (ii) the approximation error in \eqref{eq: sumJ}, both of which are incurred when approximating an autoregressive Markov process with a Wiener process. In other words, the approximation error is due to the discretization of the Wiener process. This discretization is analyzed by
	the Langevin dynamics  in \cite{AMH1987}. In particular, 
	$\frac{S_n}{\sigma}=\sum_{i=1}^{n}W_i\approx B_{n}$ can be regarded as an overdamped Langevin dynamics with step size $1$ to approximate the Brownian motion. The approximation error in each step remains  constant due to the unit step size.

	We first consider $\mathbb{E}[J_\beta]$. Substituting $\beta=\sigma\sqrt{eM}$ into $a=\beta/\sigma$ in Lemma~\ref{lem: brown motion}, we find $a=\sqrt{eM}$ is constant. So the distribution of $J$ in Lemma~\ref{lem: brown motion} does not change when $\sigma$ changes.  Thus, the approximation error in \eqref{eq: J} keeps invariant when $\sigma$ changes.
	
	Then, we consider \eqref{eq: sumJ}. $J_\beta$ is an approximation of $J$, and
	\begin{align}
		\label{eq:sumapprox}
		\sum_{j=1}^{J_\beta}S_j^2=\sigma^2\sum_{j=1}^{J_\beta}S_j^2/\sigma^2.
	\end{align}
	The distribution of $J$ does not change with $\sigma$, nor does the distribution of $J_\beta$. The terms $\frac{S_j}{\sigma}\sim\mathcal{N}(0,j)$ inside the sum in~\eqref{eq:sumapprox} are independent of $\sigma$. The distribution of $\sum_{j=1}^{J_\beta}S_j^2/\sigma^2$ does not change with $\sigma$. Thus, the approximation error  in  \eqref{eq: sumJ}  increases linearly with $\sigma^2$. 
	
	By Lemma~\ref{lem: brown motion}~(2), $\mathbb{E}[J^2] = 10\mathbb{E}[\int_0^TB_t^2dt]$. Recall that the approximation error  in  \eqref{eq: sumJ}  increases linearly with $\sigma^2$, thus the approximation error in $\mathbb{E}[J^2]$ also increases linearly with $\sigma^2$. 
	From \eqref{eq: L1 2}, the approximation error in $L^{EbT}(M)$ increases linearly with $\sigma^2$.

	\section{Unreliable Random Access Channels}\label{sec: generalization}
	
	In this section, we generalize our model to account for unreliability in random access channels, i.e., erasure channels. Related works such as \cite{ETCDG2021, AHWWCF2021} investigated Age of Information in unreliable channels, while optimal power allocation strategies in unreliable channel with respect to remote estimation has been considered in \cite{SCJJX2007}. However, in this section, we aim to minimize NAEE defined in \eqref{eq: E-MMSE} under oblivious and non-oblivious policies in unreliable channels.

	In the model defined in Section~\ref{sec:systemModel}, sensors can deliver packets successfully if no collisions happen. Now, we assume that packets are erased with some probability even if no collisions happen in the channel. In particular, suppose that if the channel is not in collision, the packet can be delivered with probability $1-\epsilon$, where $\epsilon$ is the channel erasure probability. We do not introduce another feedback, i.e., we assume that only collision feedback (not the full feedback) can be transmitted to sensors. Same as Section~\ref{sec:systemModel}, (active) sensors transmit packets through the slotted ALOHA algorithm \eqref{eq: slotted aloha}. For clarity of exposition, we assume that packets erasure happens in the end of every time slot (after channel collisions).

	From Section~\ref{sec: signal-dependent}, in the limit of $M$, the channel  throughput/rate is around $e^{-1}$ when $\epsilon=0$. Now, note that the channel erasure probability is $\epsilon$, which implies when no collisions occur, every packet chosen by slotted ALOHA \eqref{eq: slotted aloha} is delivered with probability $1-\epsilon$. Thus, the throughput/rate is around $e^{-1}(1-\epsilon)$. 
	In this section, we let $c(M)=e^{-1}(1-\epsilon)$.

	We first consider oblivious schemes. By a proof similar to that of Lemma~\ref{lem: Gaussian distribution}, \eqref{eq: Gaussian distribution} and \eqref{eq: obliviousL} still hold in our unreliable random access setting. Since the channel throughput/rate is around $e^{-1}(1-\epsilon)$, we use \cite[Theorem~5]{cxr2019} to obtain  the following NAAoI: 
	\begin{align}
		&\lim_{M\to\infty}J_{\epsilon}^{SAT}(M) = \frac{e}{2(1-\epsilon)}\label{eq: epsilon AoI}.
	\end{align}
	Note that \eqref{eq: Gaussian distribution} and \eqref{eq: obliviousL} still hold in erasure channels. Thus, the normalized average estimation error is computed by
	\begin{align}\label{eq: epsilon estimation}
		\lim_{M\to\infty}L_\epsilon^{SAT}(M) = \lim_{M\to\infty}L_\epsilon^{SAT}(M)\sigma^2=\frac{e\sigma^2}{2(1-\epsilon)}.
	\end{align}

	Now, we consider  non-oblivious schemes. The analysis in Section~\ref{sec: signal-dependent} can be generalized to yield the following theorem.
	\begin{theorem}\label{thm: optimalbeta2}
		Let $M$ be sufficient large. An optimal $\beta_{\epsilon}^*$ is approximately given by
		\begin{align*}
			\beta_{\epsilon}^* \approx \sigma\sqrt{eM/(1-\epsilon)},
		\end{align*}
		and
		\begin{align}\label{eq: estimate L1}
			\hat{L}_{\epsilon}^{EbT} = \frac{e}{6(1-\epsilon)}\sigma^2.
		\end{align}
	\end{theorem}
	\begin{remark}
		The new threshold in 
		Theorem~\ref{thm: optimalbeta2} is larger than that in Theorem~\ref{thm: optimalbeta}, i.e., $\beta_{\epsilon}^*\geq\beta^*$ for $0\leq\epsilon<1$. The expected number of newly active nodes is reduced. This is because (i) if a packet is erased, then the corresponding sensor is still active in the next time slot; (ii)  the channel throughput/rate is decreases to around $e^{-1}(1-\epsilon)$, not $e^{-1}$.
	\end{remark}
	\begin{remark}\label{remark: rationepsilon}
		Comparing \eqref{eq: epsilon estimation} and \eqref{eq: estimate L1}, we still have $L_\epsilon^{SAT}/\hat{L}_{\epsilon}^{EbT} \approx 3$ in erasure channels.
	\end{remark}
	\begin{proof}
		The proof of Theorem~\ref{thm: optimalbeta2} is the similar to that of Theorem~\ref{thm: optimalbeta}. The only difference is replacing $c(M)\approx e^{-1}$ with $c(M)\approx e^{-1}(1-\epsilon)$.
	\end{proof}

	\section{Numerical Results}\label{sec: numerical results}
	In this section, we verify our findings through simulations. 
	Figure~\ref{fig-age2} compares the NAEE of our proposed policy with the state of the art for $M=500$ under different $\sigma^2$. In this plot, the green (plus) curve corresponds to an optimal stationary randomized policy in which each node transmits with an optimal pre-determined probability. The performance of threshold policies like \cite{XZMMVWCUM2021, GT2012} that impose an optimal (fixed) transmission rate for each sensor also coincides with this curve, i.e, the green (plus) one. These policies do not exploit the available feedback for decision making. The purple (diamond) curve shows the performance of a standard  pseudo-Bayesian slotted ALOHA. Slotted ALOHA does use feedback, but treats all packets similarly, independent of their corresponding sample values. The red  (circle) and blue (squared) curves correspond to oblivious (age-based) policies \cite[Algorithm 1]{cxr2019} and \cite[Algorithm 2]{cxr2019}, respectively. The black (star) curve shows the performance of our proposed decentralized policy in Algorithm \ref{alg: Limit Threshold-slotted ALOHA} and the red ({\tt x}) curve shows the approximation we find in \eqref{eq: estimate L}. The gap between the two is small but increases linearly in $\sigma^2$ as discussed in Section~\ref{sec: Approximating Error Analysis}. On this plot, we have also included an oblivious and a non-oblivious centralized policy. 
	The  former  (green dashed  curve)  schedules  the  transmitter  with  the  largest  age and based on \cite[Section~III]{I.Kadota-2018} and Proposition~\ref{thm: MW policy} is optimal in the class of oblivious policies. The centralized non-oblivious policy that we have considered here  (yellow smooth curve) schedules the transmitter with  the largest  estimation error. Both centralized oblivious and non-oblivious policies are often observed to be numerically very close to the optimal.
	
\begin{figure*}[ht!]
\centering
\begin{subfigure}[b]{0.4\textwidth}
	\centering
	\begin{tikzpicture}[scale=0.6]
		\begin{axis}
			[axis lines=left,
			width=2.9in,
			height=2.6in,
			scale only axis,
			xlabel=$\sigma^2$,
			ylabel=NAEE,
			xmin=1, xmax=5,
			ymin=0, ymax=16,
			xtick={},
			ytick={},
			ymajorgrids=true,
			legend style={at={(0.01,1.08)},anchor=west},
			grid style=dashed,
			scatter/classes={
				a={mark=+, draw=black},
				b={mark=star, draw=black}
			}
			]

			\addplot[color=green, mark=+, thick]
			coordinates{(1,2.7024)(1.2,3.2470)(1.4,3.8018)(1.6,4.3158)(1.8,4.8816)(2.0,5.4119)(2.2,5.9531)(2.4,6.4997)(2.6,7.0363)(2.8,7.5872)(3.0,8.0906)(3.2,8.6785)(3.4,9.2475)(3.6,9.7877)(3.8,10.3007)(4.0,10.8310)(4.2,11.3364)(4.4,11.9222)(4.6,12.4251)(4.8,12.9751)(5.0,13.4665)
			};
			

			\addplot[color=blue, mark=square,thick]
			coordinates{(1,1.3603)(1.2,1.6308)(1.4,1.9040)(1.6,2.1729)(1.8,2.4474)(2.0,2.7182)(2.2,2.9900)(2.4,3.2616)(2.6,3.5315)(2.8,3.8067)(3.0,4.0777)(3.2,4.3461)(3.4,4.6200)(3.6,4.8974)(3.8,5.1649)(4.0,5.4322)(4.2,5.7129)(4.4,5.9783)(4.6,6.2529)(4.8,6.5250)(5.0,6.7944)
			};

			\addplot[color=red, mark=o, thick]
			coordinates{(1,1.0002)(1.2,1.2004)(1.4,1.4014)(1.6,1.6007)(1.8,1.8014)(2.0,1.9980)(2.2,2.2031)(2.4,2.4011)(2.6,2.6031)(2.8,2.7994)(3.0,3.0004)(3.2,3.2036)(3.4,3.4020)(3.6,3.6014)(3.8,3.8007)(4.0,4.0002)(4.2,4.19967)(4.4,4.4055)(4.6,4.6035)(4.8,4.8056)(5.0,5.0051)
			};

			\addplot[color=black, mark=star,thick]
			coordinates{(1,0.4994)(1.2,0.5987)(1.4,0.6991)(1.6,0.7988)(1.8,0.8981)(2.0,0.9980)(2.2,1.0983)(2.4,1.1981)(2.6,1.2982)(2.8,1.3969)(3.0,1.4970)(3.2,1.5964)(3.4,1.6977)(3.6,1.7980)(3.8,1.8962)(4.0,1.9962)(4.2,2.0958)(4.4,2.1973)(4.6,2.2959)(4.8,2.3969)(5.0,2.4956)
			};

			\addplot[color=green,  dashed, ultra thick]
			coordinates{(1,0.4910)(1.2,0.5910)(1.4,0.6915)(1.6,0.7918)(1.8,0.8916)(2.0,0.9910)(2.2,1.0928)(2.4,1.1908)(2.6,1.2929)(2.8,1.3924)(3.0,1.4930)(3.2,1.5933)(3.4,1.6936)(3.6,1.7917)(3.8,1.8932)(4.0,1.9947)(4.2,2.0932)(4.4,2.1956)(4.6,2.2952)(4.8,2.3952)(5.0,2.4956)
			};
			
			\addplot[color=yellow, smooth ,ultra thick]
			coordinates{(1,0.4910)(1.2,0.5810)(1.4,0.6815)(1.6,0.7818)(1.8,0.8816)(2.0,0.9810)(2.2,1.0828)(2.4,1.1808)(2.6,1.2829)(2.8,1.3824)(3.0,1.4830)(3.2,1.5833)(3.4,1.6836)(3.6,1.7817)(3.8,1.8832)(4.0,1.9847)(4.2,2.0832)(4.4,2.1856)(4.6,2.2852)(4.8,2.3852)(5.0,2.4856)
			};

			\addplot[color=red, mark=x, thick]
			coordinates{(1,0.4550)(1.2,0.5461)(1.4,0.6371)(1.6,0.7281)(1.8,0.8191)(2.0,0.9101)(2.2,1.0011)(2.4,1.0921)(2.6,1.1831)(2.8,1.2741)(3.0,1.3651)(3.2,1.4562)(3.4,1.5472)(3.6,1.6382)(3.8,1.7292)(4.0,1.8202)(4.2,1.9112)(4.4,2.0022)(4.6,2.0932)(4.8,2.1842)(5.0,2.2752)
			};

			\legend{ Optimal Stationary Randomized Policy,  Stationary Age-based Thinning in \cite{cxr2019}, Adaptive Age-based Thinning in \cite{cxr2019},  Error-based Thinning, Oblivious MW Policy,  Non-oblivious Greedy Policy, Estimated $L^{EbT}$ in \eqref{eq: estimate L}, }

		\end{axis}
	\end{tikzpicture}
\caption{NAEE as a function of $\sigma^2$ for various state-of-the-art schemes with $M=500$.}
\label{fig-age2}
\end{subfigure}
\begin{subfigure}[b]{0.4\textwidth}
	\centering
	\begin{tikzpicture}[scale=0.6]
		\begin{axis}
			[axis lines=left,
			width=2.9in,
			height=2.6in,
			scale only axis,
			xlabel=$\epsilon$,
			ylabel=NAEE,
			xmin=0, xmax=0.72,
			ymin=0, ymax=29,
			xtick={},
			ytick={},
			ymajorgrids=true,
			legend style={at={(0.01,1.05)},anchor=west},
			grid style=dashed,
			scatter/classes={
				a={mark=+, draw=black},
				b={mark=star, draw=black}
			}
			]

			\addplot[color=green, mark=+, thick]
			coordinates{(0,8.2208)(0.1,9.0395)(0.2,9.9613)(0.3,11.5209)(0.4,13.3807)(0.5,16.0379)(0.6,19.9634)(0.7,27.5479)
			};

			\addplot[color=blue, mark=square,thick]
			coordinates{(0, 4.079963)(0.1,4.611185)(0.2,5.212722)(0.3,5.953863)(0.4,6.949594)(.5,8.406248)(.6,10.504403)(.7,13.988824)
			};


			\addplot[color=black, mark=star,thick]
			coordinates{(0, 1.527396)(0.1, 1.688883)(0.2, 1.913276)(0.3,2.180743)(0.4, 2.564661)(.5,3.081347)(.6, 3.865449)(.7,5.184445)
			};

			\addplot[color=green,  dashed, ultra thick]
			coordinates{(0,1.4973)(0.1,1.6683)(0.2,1.8655)(0.3,2.1286)(0.4,2.4881)(0.5,2.9951)(0.6,3.7638)(0.7,4.9685)
			};
			
			\addplot[color=yellow, smooth ,ultra thick]
			coordinates{(0,1.4342)(0.1,1.6263)(0.2,1.8390)(0.3,2.1167)(0.4,2.4321)(0.5,2.8962)(0.6,3.7277)(0.7,4.9354)
			};

			\addplot[color=red, mark=x, thick]
			coordinates{(0,1.3591)(0.1,1.5102)(0.2,1.6989)(0.3,1.9416)(0.4,2.2652)(.5,2.7183)(.6,3.3979)(.7,4.5305)
			};

			\legend{Optimal Stationary Randomized Policy, Stationary Age-based Thinning in \cite{cxr2019},  Error-based Thinning, Oblivious MW Policy,  Non-oblivious Greedy Policy, Estimated $L_1^{EbT}$ in \eqref{eq: estimate L1}, }

		\end{axis}
	\end{tikzpicture}
\caption{NAEE as a function of $\epsilon$ for various state-of-the-art schemes with $M=500$.}
\label{fig-epsilon}
\end{subfigure}
\caption{NAEE of state-of-the-art schemes in reliable and unreliable channels}
\end{figure*}
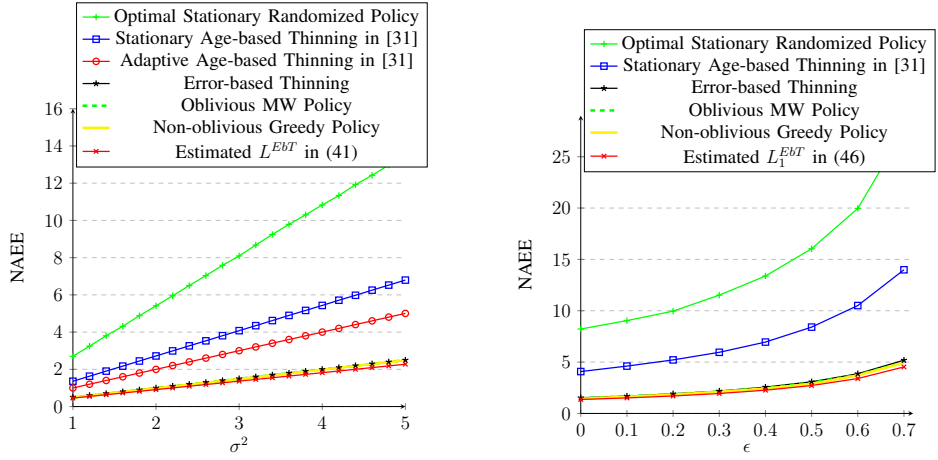	
	
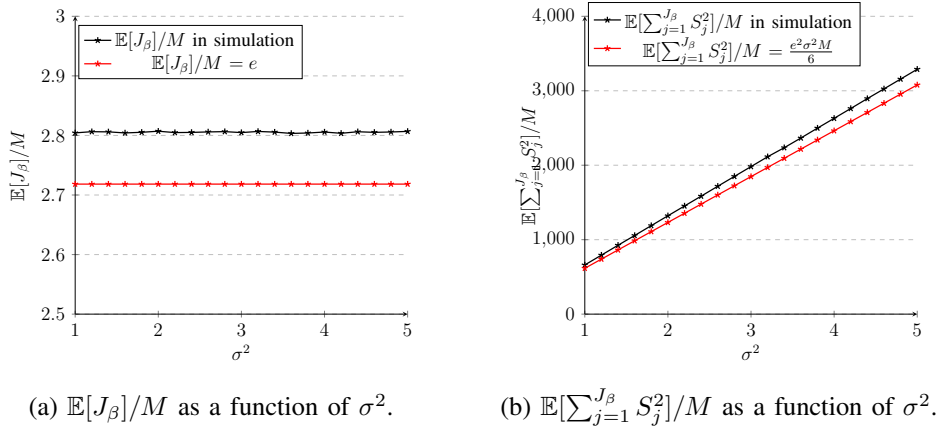
\begin{figure*}[ht!]
\centering
\begin{subfigure}[b]{0.4\textwidth}
	\centering
	\begin{tikzpicture}[scale=0.6]
		\begin{axis}
			[axis lines=left,
			width=2.9in,
			height=2.6in,
			scale only axis,
			xlabel=$\sigma^2$,
			ylabel=$\mathbb{E}{[J_\beta]}/M$ ,
			xmin=1, xmax=5,
			ymin=2.5, ymax=3,
			xtick={},
			ytick={},
			ymajorgrids=true,
			legend style={at={(0.01,0.88)},anchor=west},
			grid style=dashed,
			scatter/classes={
				a={mark=+, draw=black},
				b={mark=star, draw=black}
			}
			]

			\addplot[color=black, mark=star,thick]
			coordinates{(1,2.8043)(1.2,2.8063)(1.4,2.8060)(1.6,2.8042)(1.8,2.8053)(2.0,2.8071)(2.2,2.8048)(2.4,2.8051)(2.6,2.8059)(2.8,2.8064)(3.0,2.8048)(3.2,2.8067)(3.4,2.8056)(3.6,2.8034)(3.8,2.8041)(4.0,2.8057)(4.2,2.8036)(4.4,2.8062)(4.6,2.8052)(4.8,2.8058)(5.0,2.8068)
			};
			
			\addplot[color=red, mark=star,thick]
			coordinates{(1,2.7183)(1.2,2.7183)(1.4,2.7183)(1.6,2.7183)(1.8,2.7183)(2.0,2.7183)(2.2,2.7183)(2.4,2.7183)(2.6,2.7183)(2.8,2.7183)(3.0,2.7183)(3.2,2.7183)(3.4,2.7183)(3.6,2.7183)(3.8,2.7183)(4.0,2.7183)(4.2,2.7183)(4.4,2.7183)(4.6,2.7183)(4.8,2.7183)(5.0,2.7183)
			};

			\legend{$\mathbb{E}[J_\beta]/M$ in simulation, $\mathbb{E}[J_\beta]/M$ $=e$}
		\end{axis}
	\end{tikzpicture}
\caption{$\mathbb{E}[J_\beta]/M$ as a function of $\sigma^2$.}
\label{fig-age3}
\end{subfigure}
\begin{subfigure}[b]{0.4\textwidth}
	\centering
	\begin{tikzpicture}[scale=0.6]
		\begin{axis}
			[axis lines=left,
			width=2.9in,
			height=2.6in,
			scale only axis,
			xlabel=$\sigma^2$,
			ylabel=$\mathbb{E}{[\sum_{j=1}^{J_\beta}S_j^2]}/M$,
			xmin=1, xmax=5,
			ymin=0, ymax=4000,
			xtick={},
			ytick={},
			ymajorgrids=true,
			legend style={at={(0.01,0.94)},anchor=west},
			grid style=dashed,
			scatter/classes={
				a={mark=+, draw=black},
				b={mark=star, draw=black}
			}
			]

			\addplot[color=black, mark=star,thick]
			coordinates{(1,660.3670)(1.2,792.5328)(1.4,925.2312)(1.6,1056.1568)(1.8,1188.8486)(2.0,1320.9698)(2.2,1452.5291)(2.4,1584.7369)(2.6,1717.6587)(2.8,1849.6303)(3.0,1981.6395)(3.2,2113.8634)(3.4,2235.5457)(3.6,2365.6251)(3.8,2498.6242)(4.0,2629.8838)(4.2,2761.7586)(4.4,2895.0882)(4.6,3025.4658)(4.8,3157.2522)(5.0,3289.1200)
			};
			
			\addplot[color=red, mark=star,thick]
			coordinates{(1,615.8)(1.2,738.9)(1.4,862.1)(1.6,985.2)(1.8,1108.4)(2.0,1231.5)(2.2,1354.7)(2.4,1477.8)(2.6,1601.0)(2.8,1724.1)(3.0,1847.3)(3.2,1970.4)(3.4,2093.6)(3.6,2216.7)(3.8,2339.9)(4.0,2463.0)(4.2,2586.2)(4.4,2709.3)(4.6,2832.5)(4.8,2955.6)(5.0,3078.8)
			};

			\legend{$\mathbb{E}[\sum_{j=1}^{J_\beta}S_j^2]/M$ in simulation, $\mathbb{E}[\sum_{j=1}^{J_\beta}S_j^2]/M$ $=\frac{e^2\sigma^2M}{6}$}
		\end{axis}
	\end{tikzpicture}
	\caption{$\mathbb{E}[\sum_{j=1}^{J_\beta}S_j^2]/M$ as a function of $\sigma^2$.}
\label{fig-age4}
\end{subfigure}
\caption{The approximation errors of for $\mathbb{E}[J_\beta]/M$  and $\mathbb{E}[\sum_{j=1}^{J_\beta}S_j^2]/M$ when $M=500$}
\end{figure*}

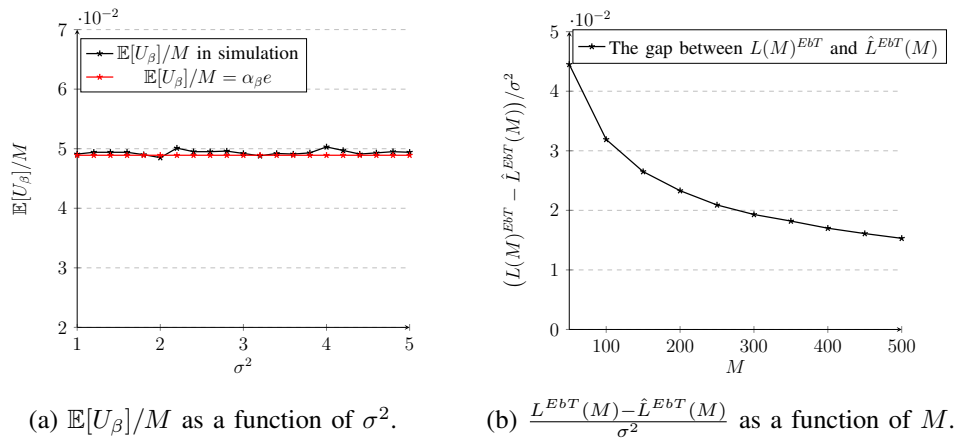
\begin{figure*}[ht!]
\centering
\begin{subfigure}[b]{0.4\textwidth}
	\centering
	\begin{tikzpicture}[scale=0.6]
		\begin{axis}
			[axis lines=left,
			width=2.9in,
			height=2.6in,
			scale only axis,
			xlabel=$\sigma^2$,
			ylabel=${\mathbb{E}[U_\beta]}/M$,
			xmin=1, xmax=5,
			ymin=0.02, ymax=0.07,
			xtick={},
			ytick={},
			ymajorgrids=true,
			legend style={at={(0.01,0.88)},anchor=west},
			grid style=dashed,
			scatter/classes={
				a={mark=+, draw=black},
				b={mark=star, draw=black}
			}
			]

			\addplot[color=black, mark=star,thick]
			coordinates{(1,0.0491)(1.2,0.0494)(1.4,0.0494)(1.6,0.0494)(1.8,0.0490)(2.0,0.0485)(2.2,0.0501)(2.4,0.0495)(2.6,0.0495)(2.8,0.0496)(3.0,0.0492)(3.2,0.0488)(3.4,0.0492)(3.6,0.0491)(3.8,0.0493)(4.0,0.0503)(4.2,0.0497)(4.4,0.0491)(4.6,0.0493)(4.8,0.0495)(5.0,0.0494)
			};
			
			\addplot[color=red, mark=star,thick]
			coordinates{(1,0.0489)(1.2,0.0489)(1.4,0.0489)(1.6,0.0489)(1.8,0.0489)(2.0,0.0489)(2.2,0.0489)(2.4,0.0489)(2.6,0.0489)(2.8,0.0489)(3.0,0.0489)(3.2,0.0489)(3.4,0.0489)(3.6,0.0489)(3.8,0.0489)(4.0,0.0489)(4.2,0.0489)(4.4,0.0489)(4.6,0.0489)(4.8,0.0489)(5.0,0.0489)
			};

			\legend{$\mathbb{E}[U_\beta]/M$ in simulation, $\mathbb{E}[U_\beta]/M=\alpha_\beta e$}
		\end{axis}
	\end{tikzpicture}
\caption{$\mathbb{E}[U_\beta]/M$ as a function of $\sigma^2$.}
\label{fig-age5}
\end{subfigure}
\begin{subfigure}[b]{0.4\textwidth}
	\centering
	\begin{tikzpicture}[scale=0.6]
		\begin{axis}
			[axis lines=left,
			width=2.9in,
			height=2.6in,
			scale only axis,
			xlabel=$M$,
			ylabel=$\big(L(M)^{EbT}-\hat{L}^{EbT}(M)\big)/\sigma^2$,
			xmin=50, xmax=500,
			ymin=0, ymax=0.05,
			xtick={},
			ytick={},
			ymajorgrids=true,
			legend style={at={(0.01,0.95)},anchor=west},
			grid style=dashed,
			scatter/classes={
				a={mark=+, draw=black},
				b={mark=star, draw=black}
			}
			]

			\addplot[color=black, mark=star,thick]
			coordinates{(50,0.0445)(100,0.0319)(150,0.0265)(200,0.0233)(250,0.0209)(300,0.0193)(350,0.0182)(400,0.0170)(450,0.0161)(500,0.0153)
			};

			\legend{The gap between $L(M)^{EbT}$ and $\hat{L}^{EbT}(M)$}
		\end{axis}
		\end{tikzpicture}
\caption{$\frac{L^{EbT}(M)-\hat{L}^{EbT}(M)}{\sigma^2}$ as a function of $M$.}
\label{fig-age6}	
\end{subfigure}
\caption{The approximation errors of for $\mathbb{E}[U_\beta]/M$ when $M=500$ (left). The gap  (normalized by $\sigma^2$) between $L^{EbT}(M)$ and $\hat{L}^{EbT}(M)$ as a function of $M$ for $\sigma^2=3$ (right).}
\end{figure*}

	The numerical calculation and analytical approximation of $\mathbb{E}[J_\beta]$, $\mathbb{E}[\sum_{j=1}^{J_\beta}S_j^2]$ and $\mathbb{E}[U_\beta]$ are given in Figure~\ref{fig-age3}, Figure~\ref{fig-age4} and Figure~\ref{fig-age5}, respectively. Recall that $\mathbb{E}[J_\beta^2]$ is $10$ times $\mathbb{E}\big[\sum_{j=1}^{J_\beta}S_j^2\big]$, so we only consider one of them. 
	In order  to offset the effect introduced by the number of nodes, we consider  the normalized silence delay $\mathbb{E}[J_\beta]/M$, the normalized transmission delay $\mathbb{E}[U_\beta]/M$, and  $\mathbb{E}[\sum_{j=1}^{J_\beta}S_j^2]/M$. The estimation error of the normalized silence delay $\mathbb{E}[J_\beta]/M$ is invariant of $\sigma^2$ (Figure~\ref{fig-age3}), while the estimation error of $\mathbb{E}[\sum_{j=1}^{J_\beta}S_j^2]/M$ increases linearly with $\sigma^2$ (Figure~\ref{fig-age4}). This coincides with the analysis in Section~\ref{sec: Approximating Error Analysis}. In the simulation, we numerically find the expected fraction of active nodes to be $\alpha_\beta = 0.0173$. Substituting $\alpha_\beta = 0.0173$ into \eqref{eq: U}, we  get $\mathbb{E}[U_\beta]$. From Figure~\ref{fig-age5}, we can see that normalized transmission delay $\mathbb{E}[U_\beta]$ coincides with analytical results in \eqref{eq: U}.

	Next, we show in Figure~\ref{fig-age6} that the gap between $L^{EbT}(M)$ and $\hat{L}^{EbT}(M)$ decreases as $M$ gets large. In other words, the influence of approximation error caused by Langevin dynamics in Algorithm \ref{alg: Limit Threshold-slotted ALOHA}  weakens (but does not vanish) as $M$ increases.

	Finally, we show the performances of different policies in unreliable random access channels with erasure probability~$\epsilon$. 
	Figure~\ref{fig-epsilon} compares the NAEE of our proposed policy with the state of the art for $M=500$ under different $\epsilon$ and $\sigma^2=3$. NAEE increases with $\epsilon$ under any policy. This is because more packets are erased when channel erasure probability is larger, hence a larger estimation error is occurred. In addition, the proposed algorithm, i.e., the EbT policy outperforms all other policies.

	\section{Conclusion and Future Work}\label{sec: Conclusion and Future Work}
	We considered the problem of real-time sampling and timely estimation over wireless collision channels with $M$ independent and statistically identical Gauss-Markov processes (sources). We studied a normalized metric of estimation error which we  termed the normalized average estimation error (NAEE), and focused on the regime of large $M$. We proposed two general classes of policies: oblivious policies and non-oblivious policies. We showed in the former class that minimizing the expected estimation error  is equivalent to minimizing the expected age and consequently provided lower and upper bounds on the optimal estimation error. We then proposed and analyzed a (non-oblivious) threshold policy in which (1) nodes become active  if their estimation error has crossed a threshold and (2) active nodes  transmit stochastically with probabilities that adapt to the state of the channel (exploiting the collision feedback).
	We showed that the NAEE performance of oblivious (age-based) policies is at least twice better than the state-of-the-art schemes (which impose a fixed rate of transmission at the nodes) such as standard slotted ALOHA and optimal stationary randomized policy. Moreover, our proposed threshold policy offers a multiplicative gain close to $3$ compared to oblivious policies. 
	Finally, we extended  our framework to incorporate unreliabile random access channels with  erasure probability $\epsilon$. The proposed optimal threshold and the corresponding NAEE increase with $\epsilon$.  Numerical results show that the multiplicative gain is $3$ and independent of $\epsilon$ (which is consistent with Remark~\ref{remark: rationepsilon}), and the additive gain (offered by  non-oblivious policies  compared to oblivious policies) increases with $\epsilon$.
	Our findings suggest that designing optimal multiple access systems for the future IoT and CPS applications requires going beyond traditional metrics of rate, reliability, latency, and age.

	Future research includes generalizations to accommodate the following scenarios: 1) dynamic networks, i.e., the number of sensors changes with time; 2) asymmetric networks, i.e., the sensors are no longer statistically identical; 3) adaptive error-based thinning policies, i.e., the threshold $\beta(k)$ changes with time $k$; 4) correlated sources, i.e., sensors are no longer mutually independent. For the first scenario,  we can simply replace $M$ with $M(k)$  in every time slot. Subsequently, the error-based threshold  is also a time-variant variable, $\beta(k)$. For the remaining three scenarios, the method we have proposed can not be applied directly. In particular, in the second scenario,  we used the profile of all the sources to find an estimate on an individual source. This step fails in asymmetric settings. In the third scenario, the nodes need statistical inference about the distribution of error process $\{\psi_i(k)\}_i$ to decide which ones are of priority.  In the fourth scenario, the policies should change to account for the correlation between the observations.

	\footnotesize
	\bibliographystyle{IEEEtran}
	\bibliography{references}

\begin{appendices}
\section{Proof of Proposition~\ref{thm: MW policy}}\label{App: Proof MW}
Recall that the proposed policy is oblivious to the monitored process. So $W_i(j)$'s are independent of $h_i(k)$.	Using \eqref{eq: Gaussian distribution},  \eqref{eq: Lyapunov}, and \eqref{eq: L Drift}, we write
\begin{align}
LD(k) =\mathbb{E}[L(k+1) -L(k)\big|\underline{T}^{(i)}(k)]=\frac{\sigma^2}{M}\sum_{i=1}^{M}\mathbb{E}\big[h_i(k+1) - h_i(k)\big].\label{eq: proof MW1}
		\end{align}
		Moreover, the age functions have the following recursion:
		\begin{align}\label{eq: proof MW2}
			h_i(k+1)=d_i(k)+(1-d_i(k))(h_i(k)+1).
		\end{align}
		where $d_i(k)\in\{0,1\}$ indicates  a successful delivery from source $i$ at time $k$. Note  $\sum_{i=1}^{M}d_i(k)=1$. Under the MW policy, no collisions occur in every time slot, so $h_i(k)$ is a scalar (not a random variable) for all $i, k$.
		Substituting $h_i(k+1)$ from \eqref{eq: proof MW2} into \eqref{eq: proof MW1}, we obtain
		\begin{align*}
			LD(k)=\frac{\sigma^2}{M}\sum_{i=1}^{M}\Big(1- h_i(k)d_i(k)\Big).
		\end{align*}
		Thus, minimizing $LD(k)$ is equivalent to choosing $i^*$ such that $h_{i^*}(k)=\max_i h_i(k)$.

		Since we assumed $h_i(0)=1$ for all nodes, from Lemma~2 in \cite[Section~III]{I.Kadota-2018}, the above MW policy  is equivalent to  a Round-Robin policy. Consequently, for all $i=1,\ldots,M$, and $k\geq i$, we get $h_i(k)=1,2,\cdots,M$  successively and periodically, and
		\begin{align*}
			\lim_{K\rightarrow\infty}\frac{1}{K}\sum_{i=1}^{M}h_i(k)=\frac{M(M+1)}{2}.
		\end{align*}
		Therefore,
		\begin{align*}
			\lim_{M\rightarrow\infty}L^{MW}(M)=\lim_{M\rightarrow\infty}\frac{\sigma^2}{M^2}\frac{M(M+1)}{2}=\frac{\sigma^2}{2}.
		\end{align*}

		\section{The Strong Law of Large Numbers holds for $\{I_\ell^{(i)}\}_\ell$}\label{App: proof of I LLN}
		From Definition~\ref{def: active nodes} and Definition~\ref{def: silience delay}, $I_\ell^{(i)}=J_\ell^{(i)}+U_\ell^{(i)}-1$, and $J_\ell^{(i)}$ is measurable and independent of $U_\ell^{(i)}$. Consider $I_1^{(i)}$ and $I_m^{(i)}$, $m\geq1$. $J_1^{(i)}$ is independent of $J_m^{(i)}$, $U_1^{(i)}$ and $U_m^{(i)}$. Then
\begin{align*}
\mathbb{E}[I_1^{(i)}I_m^{(i)}]-\mathbb{E}[I_1^{(i)}]\mathbb{E}[I_m^{(i)}] = \mathbb{E}[U_1^{(i)}U_m^{(i)}]-\mathbb{E}[U_1^{(i)}]\mathbb{E}[U_m^{(i)}]
\end{align*}
		which implies the correlation between $I_1^{(i)}$ and $I_m^{(i)}$ is the same as the correlation between $U_1^{(i)}$ and $U_m^{(i)}$.

		Now we consider the correlation between $U_1^{(i)}$ and $U_m^{(i)}$.  We first claim that the Markov process $S(k)=\big(N(k), \hat{N}(k)\big)$ is geometrically ergodic \cite{JT1987}. In fact, by Assumption~\ref{assumption: fully channel capacity}, the system is stabilized. Note that we set $\hat{\lambda}(k)=e^{-1}$ in \eqref{eq: slotted aloha} and $\lambda(k)<e^{-1}$ for all $k$. Since $\lambda_m=\limsup_k\lambda(k)<e^{-1} =\hat{\lambda}(k)$. From \cite[Theorem~3.1 and Section~IV]{JT1987}, Markov process $S(k)$ is geometrically ergodic.
		Define the state space of  $S(k)$ as $\mathcal{P}$. For any $i,j\in\mathcal{P}$, define $p_{ij}(k)=[P(k)]_{ij}$ and $\Pi=[\pi_i]_i$ as the transition probability in time slot $k$ and the stationary distribution, respectively.  A Markov chain is geometrically ergodic \cite{Spieksma1990} if there are $\rho<1$ and $C<\infty$ such that for all $i,j,k$
		\begin{align}\label{eq: geometrically ergodic}
			|p_{ij}(k)-\pi_j|\leq C\rho^k.
		\end{align}
		From \eqref{eq: geometrically ergodic}, in the limit of $k$, the transition probability equals to the stationary distribution, i.e., $\lim_{k\rightarrow\infty}p_{ij}(k)=\pi_j$ for any $i, j\in\mathcal{P}$.
		
Now, we consider $U_1^{(i)}=n$ and $U_m^{(i)}=l$. 
\begin{equation}\label{eq: EU_1U_m}
\begin{aligned}
\Pr(U_1^{(i)}=n, U_m^{(i)}=l) =\Pr(U_1^{(i)}=n)\Pr(U_m^{(i)}=l|U_1^{(i)}=n)
\end{aligned}
\end{equation}
		Define number of  the time slots between $U_1^{(i)}$ and $U_m^{(i)}$ as $m'$, $m'\geq m$. Define the states of $S(k)$ just before and after $U_1^{(i)}$ as $s_1$ and $s_2$. Define the state of $S(k)$ just before $U_m^{(i)}$ as $s_m$. In the following steps, we use $\pi_{s_i}$ and $\Pr(s_i)$ interchangeably.
		Then, due to the Markovity of $S(k)$,
		\begin{equation*}
			\begin{aligned}
			&\Pr\big(U_1^{(i)}=n, U_m^{(i)}=l\big)=\Pr(U_1^{(i)}=n)\Pr(U_m^{(i)}=l|U_1^{(i)}=n)\\
				=&\sum_{s_1,s_2,s_m\in\mathcal{P}}\Pr(s_1)\Pr(U_1^{(i)}=n|s_1)\Pr(s_2|U_1^{(i)}=n,s_1)\times\Pr(s_m|s_2)\Pr(U_m^{(i)}=l|s_m).
			\end{aligned}
		\end{equation*}
		From \eqref{eq: geometrically ergodic}, 
		\begin{align}\label{eq: pijpi}
			p_{ij}(k)=\pi_j+\epsilon_{ij}(k)
		\end{align}
		where $|\epsilon_{ij}(k)|\leq C\rho^k$ for all $i,j\in\mathcal{P}$. Note that the number of time slot between $U_1^{(i)}$ and $U_m^{(i)}$ is $m'$. By the definition of transition probabilities, 
		\begin{align*}
			\Pr(s_m|s_2)=&\sum_{s_{m-1}}p_{s_{m-1}s_m }\Pr(s_{m-1}|s_2).
		\end{align*}
		Let $\epsilon(m') = \max_{s_{m-1},s_m\in\mathcal{P}}|\epsilon_{s_{m-1}s_m}(m')|$.
		Then,
		\begin{align*}
			\Pr(s_m|s_2)\leq \big(\pi_{s_m}+\epsilon(m')\big)\sum_{s_{m-1}}\Pr(s_{m-1}|s_2)\leq\pi_{s_m}+\epsilon(m').
		\end{align*}
		Thus,
		\begin{align*}
	\Pr(U_m^{(i)}=l|U_1^{(i)}=n) \leq\sum_{s_m\in\mathcal{P}}\big(\pi_{s_m}+\epsilon(m')\big)\Pr(U_m^{(i)}=l|s_m).
		\end{align*}
		Consider the stationary distribution $\Pi$, define $$\delta=\min_i\{\pi_i>0\},$$ $\delta$ is a constant depending on the stationary distribution, hence the number of nodes $M$. Then,
		\begin{align*}
&\Pr(U_m^{(i)}=l|U_1^{(i)}=n) \leq\sum_{s_m\in\mathcal{P}}\big(\pi_{s_m}+\epsilon(m')\big)\Pr(U_m^{(i)}=l|s_m)\leq\Pr(U_m^{(i)}=l)+\epsilon(m')\sum_{s_m\in\mathcal{P}}\frac{\pi_{s_m}}{\delta}\Pr(U_m^{(i)}=l|s_m)\\
&=\Pr(U_m^{(i)}=l)+\frac{\epsilon(m')}{\delta}\sum_{s_m\in\mathcal{P}}\Pr(s_m)\Pr(U_m^{(i)}=l|s_m)=\Pr(U_m^{(i)}=l)\Big(1+\frac{\epsilon(m')}{\delta}\Big).
		\end{align*}
		Therefore,
		\begin{align*}
		\mathbb{E}[U_1^{(i)}U_m^{(i)}]\leq\big(1+\frac{\epsilon(m')}{\delta}\big)\sum_n n\Pr(U_1^{(i)}=n)\cdot\sum_l l\Pr(U_m^{(i)}=l)\leq\big(1+\frac{\epsilon(m')}{\delta}\big)\mathbb{E}[U_1^{(i)}]\mathbb{E}[U_m^{(i)}]
		\end{align*}
		Note that $\epsilon(m')\leq C\rho^{m'}$, so 
		$$
		\mathbb{E}[U_1^{(i)}U_m^{(i)}]-\mathbb{E}[U_1^{(i)}]\mathbb{E}[U_m^{(i)}]\leq \mathbb{E}[U_1^{(i)}]\mathbb{E}[U_m^{(i)}]\frac{C}{\delta}\rho^{m'}\leq C'\rho^m.
		$$
		The last equality holds because $m'\geq m$ and $\rho<1$.

		\section{Proof of Lemma~\ref{lem: alpha} }\label{App: proof of alpha} 
		Continuing from \eqref{eq:getridofi}, we have
\begin{equation}
\small
\begin{aligned}
\alpha_\beta=\mathbb{E}\left[\lim_{K\to\infty} \frac{1}{MK}\sum_{k=1}^K\sum_{i=1}^M \mathbf{1}(\text{node $i$ is active at time $k$})\right]= \mathbb{E}\left[\frac{1}{M}\sum_{i=1}^M p_a^{(i)}\right]\label{eq:getridofi_1}
			\end{aligned}
		\end{equation}
		where $p_a^{(i)}$ is the fraction of time that  node $i$ is active in the limit of $K\to\infty$,
		\begin{align}\label{eq: pai}
			p_a^{(i)}=\lim_{n\rightarrow\infty}\frac{\sum_{\ell=1}^{n}U_\ell^{(i)}}{\sum_{\ell=1}^{n}I_\ell^{(i)}}.
		\end{align}
		Furthermore,
\begin{align*}
&\mathbb{E}[p_a^{(i)}]=\mathbb{E}[\lim_{n\rightarrow\infty}\frac{\sum_{\ell=1}^{n}U_\ell^{(i)}}{\sum_{\ell=1}^{n}I_\ell^{(i)}}]=\mathbb{E}[\lim_{n\rightarrow\infty}\frac{\sum_{\ell=1}^{n}U_\ell^{(i)}/n}{\sum_{\ell=1}^{n}I_\ell^{(i)}/n}]\\&\overset{(b)}{=}\frac{\mathbb{E}[\lim_{n\rightarrow\infty}\frac{\sum_{\ell=1}^{n}U_\ell^{(i)}}{n}]}{\mathbb{E}[I_\ell^{(i)}]}\overset{(c)}{=}\frac{1}{\mathbb{E}[I_\beta]}\lim_{n\rightarrow\infty}\mathbb{E}[\frac{\sum_{\ell=1}^{n}U_\ell^{(i)}}{n}]=\frac{\mathbb{E}[U_\ell^{(i)}]}{\mathbb{E}[I_\ell^{(i)}]}.
\end{align*}
		(b) holds due to $\lim_{n\to\infty}\frac{\sum_{\ell=1}^{n}I_\ell^{(i)}}{n}=\mathbb{E}[I_\ell^{(i)}]$ in Appendix~\ref{App: proof of I LLN}. (c) holds by the dominated convergence theorem because $U_\ell^{(i)}$ is measurable. Therefore,
		\begin{align*}
			\alpha_\beta= \mathbb{E}\left[\frac{1}{M}\sum_{i=1}^M p_a^{(i)}\right]=\frac{\mathbb{E}[U_\ell^{(i)}]}{\mathbb{E}[I_\ell^{(i)}]}.
		\end{align*}

		\section{Proof of Lemma~\ref{lem EU}}\label{App: proof of EU EI}
		
		{\bf (1)} Note that the channel throughput is $c(M)$. Define $n_i$ as the total delivered number of packets delivered from node $i$ up to and including time slot $K$. Note that the transmission policy is stationary, so $n_i\to\infty$ implies $K\to\infty$. By Appendix~\ref{App: proof of I LLN}, the Law of Large Number holds for $\{I_\ell^{(i)}\}$, so the throughput is
		\begin{align*}
			\sum_{i=1}^{M}\lim_{n_i\rightarrow\infty}\frac{n_i}{\sum_{\ell=1}^{n_i}I_\ell^{(i)}}=\frac{M}{\mathbb{E}[I_\beta]}=c(M),
		\end{align*}
		which implies 
		\begin{align*}
			\mathbb{E}[I_\beta]=\frac{M}{c(M)}.
		\end{align*}

		{\bf (2)} Using  Lemma~\ref{lem: alpha}, we then obtain $\mathbb{E}[U_\ell^{(i)}]=\frac{M}{c(M)}\alpha_\beta$, hence $\mathbb{E}[U_\beta]=\frac{M}{c(M)}\alpha_\beta$. 
		From \eqref{eq: alpha_solution}, noting that $\alpha_\beta\leq1$, we find
		\begin{align*}
			\alpha_\beta=\frac{1-\sqrt{1-4\lambda_m/M}}{2} = \frac{2\lambda_m/M}{1+\sqrt{1-4\lambda_m/M}}\leq\frac{2\lambda_m}{M}.
		\end{align*}
		Therefore, we have 
		\begin{align*}
			M\alpha_\beta\leq 2\lambda_m< 2e^{-1}.
		\end{align*}	
		From Assumption~\ref{assumption: fully channel capacity}, under an optimal $\beta$, if $M$ is sufficient large, then $c(M)\approx e^{-1}$, so there exists $\epsilon>0$, such that $c(M)>e^{-1}-\epsilon$. Therefore, from \eqref{eq: U}, we find
		\begin{align*}
			\mathbb{E}[U_\beta] = o(M).
		\end{align*}

		\section{Proof of Lemma~\ref{lem: I delta SSN}}\label{App: LLN for sequences}
		From \eqref{eq: def delta}, \eqref{eq: E-MMSE-1} can be written as
		\begin{align}\label{LDelta}
			L^{EbT}(M)=\lim_{K\rightarrow\infty}\mathbb{E}[\frac{1}{M^2K}\sum_{i=1}^{M}\sum_{\ell=1}^{n_i}\Delta_\ell^{(i)}]
		\end{align}
		where $n_i$ is the total number of packets delivered from source $i$ up to and including time slot $K$.
		From the proof of Appendix~\ref{App: proof of I LLN},  $\{I_\ell^{(i)}\}$ is measurable.  Then, from \eqref{eq: def delta}, $\Delta_\ell^{(i)}$ is measurable. 
		By the dominated convergence theorem, we can exchange the order of $\lim_{K\to\infty}$ and $\mathbb{E}$ in \eqref{LDelta}.
		
		Note that $\{K\rightarrow\infty\}$ is equivalent to $\{n_i\rightarrow\infty\}$ for all $i$.
		It follows that in the limit of large time horizon $K$ (equivalently, large $n_i$ for all $i$), we have
		\begin{align*}
			L^{EbT}(M)=&\mathbb{E}[\frac{1}{M^2}\sum_{i=1}^{M}\lim_{n_i\rightarrow\infty}\sum_{\ell=1}^{n_i}\frac{\Delta_\ell^{(i)}}{I_\ell^{(i)}}]=\frac{1}{M^2}\sum_{i=1}^{M}\mathbb{E}\Bigg[\lim_{n_i\rightarrow\infty}\frac{\sum_{\ell=1}^{n_i}\Delta^{(i)}_\ell/n_i}{\sum_{\ell=1}^{n_i}I_\ell^{(i)}/n_i}\Bigg]\\
			=&\frac{1}{M^2}\sum_{i=1}^{M}\frac{1}{\mathbb{E}[I_\ell^{(i)}]}\lim_{n_i\rightarrow\infty}\mathbb{E}\Bigg[\frac{\sum_{\ell=1}^{n_i}\Delta^{(i)}_\ell}{n_i}\Bigg]=\frac{1}{M}\frac{\mathbb{E}[\Delta^{(i)}_\ell]}{\mathbb{E}[I_\ell^{(i)}]}.
		\end{align*}
		The last equality holds because $\Delta_\ell^{(i)}$ is identical over $\ell$. Recall that $\Delta_\beta$ and $I_\beta$ have the same distribution as $\Delta_\ell^{(i)}$ and $I_\ell^{(i)}$, respectively. Therefore,
		\begin{align*}
			L^{EbT}(M)=\frac{1}{M}\frac{\mathbb{E}[\Delta_\beta]}{\mathbb{E}[I_\beta]}.
		\end{align*}

		\section{Proof of \eqref{eq: Lpi222}}\label{App: proof of L2}
		For any $J_\beta+1\leq j\leq J_\beta+U_\beta-1$, $W_j$ is independent of $S_{J_\beta}$, hence $J_{\beta}$. Therefore, 
		\begin{align*}
	\mathbb{E}[S_j^2]=\mathbb{E}[(S_{J_\beta}+W_{J_{\beta}+1}+\cdots+W_j)^2]=\mathbb{E}[S_{J_\beta}^2]+\mathbb{E}[j-J_{\beta}]\sigma^2.
		\end{align*}
		Note that given $\beta$, $J_{\beta}$ and $U_{\beta}$ are independent, This helps to further simplify the numerator of $L_2^{EbT}(M)$ in \eqref{eq: Lpi222},
		\begin{align*}
			&\mathbb{E}\big[\sum_{j=J_\beta+1}^{J_\beta+U_\beta-1}S_j^2\big]=\mathbb{E}_{U_{\beta}}\Big\{\mathbb{E}\big[\sum_{j=J_\beta+1}^{J_\beta+U_\beta-1}(S_{J_\beta}+W_{J_{\beta}+1}+\cdots+W_j)^2|U_{\beta}\big]\Big\}\\
			=&\mathbb{E}_{U_{\beta}}\Big\{(U_{\beta}-1)\mathbb{E}[S_{J_{\beta}}^2]+\mathbb{E}[\sum_{j=J_{\beta}+1}^{J_{\beta}+U_{\beta}-1}(j-J_{\beta})]\sigma^2|U_\beta\Big\}=\mathbb{E}_{U_{\beta}}\Big\{(U_{\beta}-1)\mathbb{E}[S_{J_{\beta}}^2]+\frac{U_{\beta}(U_{\beta}-1)}{2}\sigma^2|U_\beta\Big\}\\
			=&(\mathbb{E}[U_{\beta}]-1)\mathbb{E}[S_{J_{\beta}}^2]+\mathbb{E}\big[\frac{U_{\beta}(U_{\beta}-1)}{2}\big]\sigma^2.
		\end{align*}
		Substituting \eqref{eq: EJ} into the equation above,
		\begin{align*}
			L_2^{EbT}(M)=\frac{1}{M}\cdot\frac{2\mathbb{E}[J_\beta](\mathbb{E}[U_\beta]-1)+\mathbb{E}[U_\beta^2]-\mathbb{E}[U_\beta]}{2\mathbb{E}[I_\beta]}\sigma^2.
		\end{align*}

		\section{Proof of Lemma~\ref{lem: brown motion}}\label{App: Proof of brown motion}
		The proof of the first part is the same as that of Theorem~7.5.5 and Theorem~7.5.9 in  \cite[Chapter~7]{Durrett}. 
		Here, we prove the second part. Using  \cite[Lemma~4]{YSun-2017}, we have
		\begin{align*}
			\mathbb{E}[\int_{0}^{J}B_t^2dt] = \frac{1}{6}\mathbb{E}[B_J^4].
		\end{align*}
		From the definition of $J$, $B_J^4=a^4$, then $\mathbb{E}[B_J^4]=a^4$, hence
		\begin{align*}
			\mathbb{E}[\int_{0}^{J}B_t^2dt]=\frac{1}{6}a^4.
		\end{align*}
		From Theorem~7.5.9 in  \cite[Chapter~7]{Durrett}, $\mathbb{E}[J^2]=\frac{5}{3}a^4$, so $\mathbb{E}[\int_{0}^{J}B_t^2dt] = \frac{1}{10}\mathbb{E}[J^2]$.

		\section{Proof of Theorem~\ref{thm: optimalbeta}}\label{App: Proof of optimal beta}
		We start with the expression of $\hat{L}^{EbT}(M)$ in \eqref{eq: L1 2}. Using   \eqref{eq: IlLlUl}, \eqref{eq: L1 2} can be re-written as
		\begin{align*}
	&\frac{1}{\sigma^2}\hat{L}^{EbT}(M)= \frac{\frac{1}{5}\mathbb{E}[J_\beta^2] + \mathbb{E}[U_\beta^2]}{2M\mathbb{E}[I_\beta]} =\frac{\frac{1}{5}\mathbb{E}[J_\beta^2] + \mathbb{E}[(I_\beta-J_\beta-1)^2]}{2M\mathbb{E}[I_\beta]}\\
			&=\frac{\frac{1}{5}\mathbb{E}[J_\beta^2] + \mathbb{E}[I_\beta^2]+\mathbb{E}[J_\beta^2] +1 - 2\mathbb{E}[I_\beta] +2\mathbb{E}[J_\beta]- 2\mathbb{E}[I_\beta J_\beta]}{2M\mathbb{E}[I_\beta]}
		\end{align*}
		Now replace for $I_\beta$ in $\mathbb{E}[I_\beta J_\beta]$ using \eqref{eq: IlLlUl}. Consider $M$  sufficiently large, and note that  $J_\beta\leq I_\beta$. We can approximately write the equation above as follows
		\begin{align}
			\hat{L}^{EbT}(M)\approx&\frac{\frac{1}{5}\mathbb{E}[J_\beta^2] + \mathbb{E}[I_\beta^2] - 2\mathbb{E}[(U_\beta+1) J_\beta] -\mathbb{E}[J_\beta^2]}{2M\mathbb{E}[I_\beta]}\sigma^2\nonumber
			\stackrel{(a)}{=}\frac{\frac{1}{5}\mathbb{E}[J_\beta^2] + \mathbb{E}[I_\beta^2] - 2\mathbb{E}[U_\beta+1]\mathbb{E}[ J_\beta] -\mathbb{E}[J_\beta^2]}{2M\mathbb{E}[I_\beta]}\sigma^2\nonumber\\
			\stackrel{(b)}{\approx}&\frac{\frac{1}{5}\mathbb{E}[J_\beta^2] + \mathbb{E}[I_\beta^2] -\mathbb{E}[J_\beta^2]}{2M\mathbb{E}[I_\beta]}\sigma^2
			\stackrel{(c)}{=}\frac{\frac{1}{5}\mathbb{E}[J_\beta^2] + \big(\mathbb{E}[I_\beta]\big)^2 + Var(U_\beta) -  \big(\mathbb{E}[J_\beta]\big)^2}{2M\mathbb{E}[I_\beta]}\sigma^2\label{eq:lbtapprox}
		\end{align}
		where $(a)$ holds because $U_\beta$ and $J_\beta$ are independent  given $\beta$, $(b)$ holds because $\mathbb{E}[U_\beta] = o(M)$ (see \eqref{eq: U}) and $J_\beta<I_\beta$, and $(c)$ holds by \eqref{eq: IlLlUl} and the independence of $U_\beta$ and $J_\beta$ which leads to $Var(U_\beta) + Var(J_\beta) = Var(I_\beta)$.
		
		Substituting \eqref{eq: I_beta} and \eqref{eq: J} into \eqref{eq:lbtapprox}, we obtain
		\begin{align}\label{eq: LbetaU}
			\hat{L}^{EbT}(M)\approx \frac{\frac{M^2}{c(M)^2} - \frac{2\beta^4}{3\sigma^4}+Var(U_\beta)}{2\frac{M^2}{c(M)}}\sigma^2.
		\end{align}
		
		Note that $J_\beta$, as defined before, is a stopping time of the discretization of the considered Wiener process $B(t)$, and therefore $J_\beta> J$, almost everywhere.
		We thus conclude that
		\begin{align}\label{eq: lower bound J}
			\mathbb{E}[J_\beta]>\mathbb{E}[J].
		\end{align}
		Delay per transmission is $1$ time slot, so $U_\beta\geq1$.  
		Using  \eqref{eq: I_beta} and \eqref{eq: lower bound J}, we can write
		\begin{align}\label{eq:beta}
			0\leq \mathbb{E}[U_\beta]-1 =\frac{M}{c(M)}-\mathbb{E}[J_\beta]< \frac{M}{c(M)}-\mathbb{E}[J].
		\end{align}
		Substituting $ \mathbb{E}[J]=\frac{\beta^2}{\sigma^2}$ (see \eqref{eq: J}) into \eqref{eq:beta}, we find
		\begin{align*}
			\beta\leq\sigma\sqrt{\frac{M}{c(M)}}.
		\end{align*}
		It is now easy to see that $\hat{\beta}=\sigma\sqrt{\frac{M}{c(M)}}$ is the minimum point of $\frac{- \frac{2\beta^4}{3\sigma^4}}{2\frac{M^2}{c(M)}}$.
		We will next show that the term $\frac{Var(U_{\beta^*})}{2\frac{M^2}{c(M)}}$ (in \eqref{eq: LbetaU}) is negligible and therefore $\beta^*\approx\hat{\beta}$ is approximately optimal. This will lead to $\hat{L}^{EbT} = \frac{1}{6c(M)}\sigma^2$.

		Recall from \cite[Theorem~3.1 and Section~IV]{JT1987} (see also Appendix~\ref{App: proof of I LLN}) that $S(k)  = \big(N(k),\hat{N}(k)\big)$ (depending on $\beta^*$) is geometrically ergodic.  Let $$S(k)\in\{(e_1,s_1), (e_2,s_2), \cdots, (e_m,s_m)\},$$ and $m$ be finite. Define the state space of $S(k)$ as $\mathcal{P}$, $|\mathcal{P}|=m$. Let $p_{ij}(k)$ be the transition probability matrix in time slot $k$, $i, j\in\mathcal{P}$. Let $\pi_i$ be the stationary distribution of $S(k)$. For $i=1,2,\cdots,m$, the transmitting probability $\nu_i$ is obtained by \eqref{eq: slotted aloha}, and the corresponding probability of a successful delivery of {\it each} active node, denoted by $r_i$, is 
		\begin{align*}
			r_i = \left\{
			\begin{aligned}
				&\nu_i(1-\nu_i)^{e_i-1}&\quad& e_i\geq 1\\
				&0&\quad&e_i = 0
			\end{aligned}
			\right.
		\end{align*}
		for $i=1,2,\cdots,m$. 
		Denote by $\Pr(i,k)$  the probability that the system is in state $i$ in time slot $k$. Let $H(j, i, k)\in\{0,1\}$ be an indicator. $H(j, i, k)=1$ represents that node $j$ becomes newly active in time slot $k$ when the system is in state $i$. Consider any {\it node} $j$. We have
		\begin{align*}
			\Pr(U_{\beta^*}=1) =&\lim_{k\to\infty}\sum_{i=1}^{m}\Pr(i,k)\mathbb{E}[H(j, i,k)]r_i
		\end{align*}
		Note that the system is stationary, and $\lim_{k\to\infty}\mathbb{E}[H(j, i,k)]$ exists for all $j, i$.
		Since all nodes are identical, when the system is stationary, we have 
		\begin{align*}
			\lim_{k\to\infty}\mathbb{E}[H(1, i,k)] = \cdots = \lim_{k\to\infty}\mathbb{E}[H(M, i,k)] \triangleq \tilde{h}_i.
		\end{align*}
		In addition, when the system is stationary, 
		\begin{align*}
			\lim_{k\to\infty}\Pr(i,k) = \pi_i.
		\end{align*}
		Denote the dominant term of $\Pr(U_{\beta^*}=1)$ as $\Delta_1$, thus
		\begin{align*}
			\Delta_1 = \sum_{i=1}^m\pi_i \tilde{h}_i r_i.
		\end{align*}
		Similarly, 
		\begin{align*}
			\Pr(U_{\beta^*}=2)=&\lim_{k\to\infty}\sum_{i=1}^{m}\Pr(i, k)\mathbb{E}[H(j, i, k)] (1-r_i)\times\sum_{l=1}^{m}p_{il}(k+1)r_l.
		\end{align*}
		Since $S(k)$ is geometrically ergodic, we have $|p_{ij}(k) - \pi_j|\leq C\rho^k$, where $C<\infty$, $0<\rho<1$.
		Thus, the dominant term of $\Pr(U_{\beta^*}=2)$ when  $k\to\infty$, denoted by $\Delta_2$, is
		\begin{align*}
			\Delta_2=\sum_{i=1}^m\pi_i\tilde{h}_i(1-r_i)\sum_{j=1}^m\pi_j r_j = (y-\Delta_1)\mu
		\end{align*}
		where $y=\sum_{i=1}^m\pi_i\tilde{h}_i$ and $\mu = \sum_{j=1}^m\pi_j r_j$. 
		By a similar process, denote the dominant term of $\Pr(U_{\beta^*}=l)$ as $\Delta_l$:
		\begin{align*}
			\Delta_l = (\mu-\Delta_1)(1-\mu)^{l-2}\mu,\quad l\geq3.
		\end{align*}
		So the dominant term of $\mathbb{E}[U_{\beta^*}]$, denoted by $\Lambda_1$, is
		\begin{align*}
			\Lambda_1 = \Delta_1+(y-\Delta_1)\mu\sum_{l=2}^{\infty}(1-\mu)^{l-2}l=\Delta_1 + 2(y-\Delta_1) + \frac{1-\mu}{\mu}(y-\Delta_1)=y+(y-\Delta_1)\frac{1}{\mu}.
		\end{align*}
		By Lemma~\ref{lem EU}, $\mathbb{E}[U_{\beta^*}] = o(M)$, so $\Lambda_1 = o(M)$. Note that $\frac{1}{\mu} = \frac{\Lambda_1-y}{y-\Delta_1}$, then $\frac{1}{\mu} = o(M)$ since $y$ and $\Delta_1$ are scalars.
		
		The dominant term of $\mathbb{E}[U_{\beta^*}^2]$, denoted by $\Lambda_2$, is similarly
		\begin{align*}
			\Lambda_2 =& \Delta_1+(y-\Delta_1)\mu\sum_{l=2}^{\infty}(1-\mu)^{l-2}l^2= \Delta_1+(y-\Delta_1)\mu\sum_{l=2}^{\infty}(1-\mu)^{l-2}\big((l-2)^2 + 4(l-2) -4\big)\\
			=&\Delta_1 + (y-\Delta_1)\big(\frac{(1-\mu)(2-\mu)}{\mu^2}+4\frac{1-\mu}{\mu} - 1\big).
		\end{align*}
		Note that $\frac{1}{\mu} = o(M)$, thus $\Lambda_2 = o(M^2)$,  $\mathbb{E}[U_{\beta^*}^2] = o(M^2)$, and
		\begin{align*}
			Var(U_{\beta^*}) = o(M^2)
		\end{align*}
		which implies
		\begin{align*}
			\frac{Var(U_{\beta^*})}{M^2}\approx0.
		\end{align*}
		
		So $\beta^*\approx\hat{\beta}=\sigma\sqrt{\frac{M}{c(M)}}$ is approximately optimal and $\hat{L}^{EbT} \approx \frac{1}{6c(M)}\sigma^2$. From Assumption~\ref{assumption: fully channel capacity}, when $M$ is sufficiently large, $c(M)\approx e^{-1}$, then $\hat{L}^{EbT}\approx \frac{e}{6}\sigma^2$, and the corresponding $\beta^* \approx\hat{\beta}= \sigma\sqrt{eM}$.

		\section{Assumptions~\ref{assumption: independent}, ~\ref{assumption: fully channel capacity}  are (approximately) satisfied}\label{App: assumptions verify}
		We first verify Assumption~\ref{assumption: fully channel capacity}. From the proof of Theorem~\ref{thm: optimalbeta}, when $M$ is sufficiently large,  $\mathbb{E}[J_{\beta^*}] \approx \mathbb{E}[I_{\beta^*}]$, and according to \eqref{eq: IlLlUl}, thus $\mathbb{E}[U_{\beta^*}]\approx 1$.  Substituting $\mathbb{E}[U_{\beta^*}]\approx 1$ into  \eqref{eq: alpha_solution}, $\lambda_m \approx (1-\alpha_{\beta^*})c(M)<c(M)<1/e$, which implies the system is stabilized. 
		In addition, $\mathbb{E}[J_{\beta^*}]\approx eM$, which implies on average (approximately) every $eM$ slots a source becomes newly active. 
		Note that $\lambda_m=\lim_{k\to\infty}\lambda(k) \approx\frac{1}{eM}\times M=e^{-1}$. 
		From \eqref{eq: alpha_solution},
		\begin{align*}
			(1-\alpha_{\beta^*})M\alpha_{\beta^*} \approx e^{-1}.
		\end{align*}
		Solving the equation for $\alpha_{\beta^*}\leq1$ when $M$ is sufficiently large, we find
		\begin{align*}
			\alpha_{\beta^*}\approx \frac{1}{eM}.
		\end{align*}
		
		Recall that $N(k)$ is the number of active nodes in the end of time slot $k$. Denote $D(k)$ as the number of nodes that become inactive in time slot $k$. Then, $\lim_{k\to\infty}\mathbb{E}[D(k)] = c(M)$. Thus,
		\begin{equation}\label{eq: verify N(k)M}
			\begin{aligned}
				N(k+1) = \min\{N(k) + \lambda(k+1), M\} - D(k+1).
			\end{aligned}
		\end{equation}
		From \eqref{eq: limit of alpha(k)} and Lemma~\ref{lem: alpha}, $\lim_{k\to\infty}\mathbb{E}[N(k)] = \alpha_{\beta^*} M$. Under $\beta^*$, when $M$ is sufficiently large, from \eqref{eq: U}, $\lim_{k\to\infty}\mathbb{E}[N(k)] = \mathbb{E}[U_{\beta^*}]c(M) \approx c(M) <1$. Therefore, from \eqref{eq: verify N(k)M},
		\begin{align}\label{eq: verify N(k)M1}
			\mathbb{E}[N(k+1)] = \mathbb{E}[N(k)] + \mathbb{E}[\lambda(k+1)] - \mathbb{E}[D(k+1)].
		\end{align}
		Letting $k\to\infty$, \eqref{eq: verify N(k)M1} is reduced to
		\begin{align*}
			c(M) = \lambda_m.
		\end{align*}
		For sufficiently large $M$, $c(M)\approx e^{-1}$ and
		Assumption~\ref{assumption: fully channel capacity} is (approximately) satisfied.

		Next, we verify Assumption~\ref{assumption: independent}. 
		Denote $q(n,L)$ as the probability that $n$ of $L$ inactive nodes becomes newly active in one slot. Since the system is stationary,  $q(n,L)$ will not change over time and only depends on the error process profile. Recall that $a(k)$ is number of newly active nodes in time $k$. We need to show $\{a(k)=l_1\}$ and $\{a(k+1)=l_2\}$ are independent where $l_1, l_2$ are non-negative integers. In fact,
		\begin{align*}
\Pr\{a(k)=l_1, a(k+1)=l_2\} = \Pr\{a(k)=l_1\}\Pr\{a(k)=l_2|a(k)=l_1\}
		\end{align*}
		Note that $\sum_{i=1}^{M}d_i(k)=1$ represents that a packet is delivered in time slot $k$. Then,
		\begin{align*}
			&\Pr\{a(k)=l_2|a(k)=l_1\}=\Pr\{a(k)=l_2|a(k)=l_1, \sum_{i=1}^{M}d_i(k)=0\}\Pr\{\sum_{i=1}^{M}d_i(k)=0\}\\
			&+\Pr\{a(k)=l_2|a(k)=l_1, \sum_{i=1}^{M}d_i(k)=1\}\Pr\{\sum_{i=1}^{M}d_i(k)=1\}
		\end{align*}
		Note that $c(M)\approx e^{-1}$, then $\Pr\{\sum_{i=1}^{M}d_i(k)=1\}\approx e^{-1}$, thus
		\begin{align*}
	\Pr\{a(k)=l_2|a(k)=l_1\}\approx q(l_2,M-l_1-1)1/e+q(l_2,M-l_1)(1-1/e).
		\end{align*}
		Recall that $\mathbb{E}[a(k)]=\lambda_m$ when $k\to\infty$, i.e., the system is stationary. Note that $a(k)$ is non-negative, so by Markov's Inequality, we have
		\begin{align*}
			\Pr\{a(k)\geq O(M)\}\leq\frac{\mathbb{E}[a(k)]}{O(M)}.
		\end{align*}
		Let $k, M\to\infty$, we have $\Pr\{a(k)\geq O(M)\}\to0$, which implies $a(k)=o(M)$ with probability $1$ when $M$ is sufficiently large.  $q(l_2,M-l_1-1)\approx q(l_2,M-l_1)\approx q(l_2, M)\approx \Pr(a(k+1)=l_2)$, thus
		\begin{align*}
\Pr\{a(k)=l_1, a(k+1)=l_2\}\approx \Pr\{a(k)=l_1\}\Pr\{a(k)=l_2\}
		\end{align*}
		when $M$ is sufficiently large. Assumption~\ref{assumption: independent} is thus (approximately) satisfied.

	\end{appendices}

\end{document}